\documentclass[11pt, reqno]{amsart}
\usepackage{amsfonts,latexsym,enumerate}
\usepackage{amsmath}
\usepackage{amscd}
\usepackage{float,amsmath,amssymb,mathrsfs,bm,multirow,graphics}
\usepackage[dvips]{graphicx}
\usepackage[percent]{overpic}
\usepackage[pdftex]{color}
\usepackage{amsaddr}
\usepackage[numbers,sort&compress]{natbib}

\newcommand{\R}{{\Bbb R}}

\newcommand{\C}{{\Bbb C}}

\newcommand{\re}{\text{\upshape Re\,}}

\newcommand{\im}{\text{\upshape Im\,}}

\newcommand{\ntlim}{\lim^\angle}

\def\XXint#1#2#3{{\setbox0=\hbox{$#1{#2#3}{\int}$}
\vcenter{\hbox{$#2#3$}}\kern-.5\wd0}}

\newtheorem{theorem}{Theorem}[section]

\newtheorem{lemma}[theorem]{Lemma}

\newtheorem{assumption}[theorem]{Assumption}
\newtheorem{remark}[theorem]{Remark}

\newtheorem{figuretext}{Figure}

\numberwithin{equation}{section}

\input epsf
\date{\today}
\title[Long-time asymptotics for the derivative nonlinear Schr\"odinger]
{Long-time asymptotics for the derivative nonlinear Schr\"odinger equation on the half-line}

\author{Lynnyngs Kelly Arruda}
\address{Department of Mathematics, Federal University of S\~ao Carlos, \\ PO B 676, 13565-905 S\~ao Carlos, Brazil.}
\email{lynnyngs@dm.ufscar.br, lynnyngs@kth.se}

\author{Jonatan Lenells}
\address{Department of Mathematics, KTH Royal Institute of Technology, \\ 100 44 Stockholm, Sweden }
\email{jlenells@kth.se}

\begin{document}
\begin{abstract} 
\noindent
We derive asymptotic formulas for the solution of the derivative nonlinear Schr\"odinger equation on the half-line under the assumption that the initial and boundary values lie in the Schwartz class. The formulas clearly show the effect of the boundary on the solution. The approach is based on a nonlinear steepest descent analysis of an associated Riemann-Hilbert problem. 
\end{abstract}

\maketitle



\setcounter{tocdepth}{1}
\tableofcontents

\section{Introduction}
The derivative nonlinear Schr\"odinger (DNLS) equation
\begin{align}\label{dnls}
iu_t +  u_{xx} = i (|u|^2u)_x, 
\end{align}
arises as a model both in nonlinear fiber optics and in plasma physics. 
In the context of fiber optics, it models the propagation of nonlinear pulses in optical fibers when certain higher-order nonlinear effects, such as self-steepening, are taken into account \cite{K1985, A2007}. In the context of plasma physics, it describes the motion of Alfv\'en waves propagating parallel to the ambient magnetic field \cite{M1976}. 

In addition to its physical relevance, equation (\ref{dnls}) has several mathematically appealing properties stemming from its integrability. In \cite{KN1978}, Kaup and Newell presented a Lax pair for equation (\ref{dnls}) and implemented the inverse scattering transform for the solution of the initial-value problem (IVP) in the case of decaying initial data. An analogous implementation of the inverse scattering transform for solutions $u(x,t)$ which approach a constant as $x\to \pm \infty$ was presented in \cite{KI1978}. Equation (\ref{dnls}) admits solitons with exponential decay at infinity as well as algebraic solitons \cite{KN1978}. 
Wellposedness of the IVP for (\ref{dnls}) was studied in \cite{HO1992}, both in Sobolev spaces and in the Schwartz class. Global existence results can also be obtained via inverse scattering techniques \cite{LPS2016}. The long-time asymptotics of the solution of the IVP can be determined by performing a nonlinear steepest descent analysis of the associated Riemann-Hilbert (RH) problem \cite{KV1997, XF2012, XFC2013}.

In this paper, we are concerned with the initial-boundary value problem (IBVP) for equation (\ref{dnls}) posed in the quarter-plane domain 
$$\{(x,t) \in \R^2 \, | \, x \geq 0, t \geq 0\}.$$
This IBVP can be analyzed by means of the unified transform for nonlinear integrable PDEs introduced in \cite{F1997}, which is a generalization of the inverse scattering transform to the setting of IBVPs. In fact, assuming that the solution is sufficiently smooth and has rapid decay as $x \to \infty$, it was shown in \cite{L2008} that $u(x,t)$ can be represented in terms of the solution of a matrix RH problem with jump matrices given in terms of four spectral functions $a(k)$, $b(k)$, $A(k)$, $B(k)$. The functions $a(k)$, $b(k)$ are defined in terms of the initial data $u_0(x) = u(x,0)$, whereas $A(k)$, $B(k)$ are defined in terms of the boundary values $g_0(t) = u(0,t)$ and $g_1(t) = u_x(0,t)$. 

Our purpose here is to derive the long-time asymptotics of the solution $u(x,t)$ on the half-line by performing a nonlinear steepest descent analysis of the RH problem of \cite{L2008}. 
We assume that the initial and boundary values lie in the Schwartz class. 
Our main result (Theorem \ref{mainth}) shows that the long-time asymptotics of $u(x,t)$ in the sector defined by (see Figure \ref{sector.pdf})
\begin{align}\label{similaritysector}
0 \leq x \leq N t, \qquad \text{$N$ constant},
\end{align}
is given by
\begin{align}\label{uasymptoticsintro}
u(x,t) = \frac{u_a(x,t)}{\sqrt{t}} + O\bigg(\frac{\ln t}{t}\bigg), \qquad t \to \infty, \ 0 \leq x \leq Nt,
\end{align}
where $u_a(x,t)$ is given explicitly in terms of the `reflection coefficient' $r(\lambda)$ defined by
$$r(\lambda) = \frac{\overline{b({\bar k})}}{ka(k)} - \frac{ \overline{B(\bar{k})}}{k a(k) d(k)}, \qquad \lambda = k^2 \in (-\infty,0),$$
with $d(k) := a(k)\overline{A(\bar{k})} -  b(k) \overline{B(\bar{k})}$. Of particular interest is the influence of the boundary on the long-time asymptotics. 
It turns out that the global relation (a relation between the initial and boundary values expressed at the level of the spectral functions) forces the spectral function $r(\lambda)$ to vanish to all orders at $\lambda = 0$. This, in turn, implies that the leading-order coefficient $u_a(x,t)$ in (\ref{uasymptoticsintro}) satisfies $u_a(x,t) = O(|x/t|^n)$ for every $n \geq 1$ as $x/t \to 0$ (see Remark \ref{effectremark}). The fact that $u_a(x,t)$ vanishes at $x = 0$ is consistent with our assumption of rapidly decaying boundary values. The vanishing of $u_a(x,t)$ as $x/t \to 0$ also illustrates the effect of the boundary on the solution in the region near $x = 0$. Indeed, for the pure IVP problem (i.e., in the absence of a boundary), the asymptotics of the solution $u(x,t)$ is given by a formula similar to (\ref{uasymptoticsintro}) except that the spectral function $r(\lambda)$ is defined in terms of the initial data alone. In this case, there is no global relation which forces $r(\lambda)$ to vanish at the origin. Hence, in the absence of a boundary, the solution, in general, only decays like $t^{-1/2}$ near $x = 0$.

\begin{figure}
\bigskip
\begin{center}
\begin{overpic}[width=.45\textwidth]{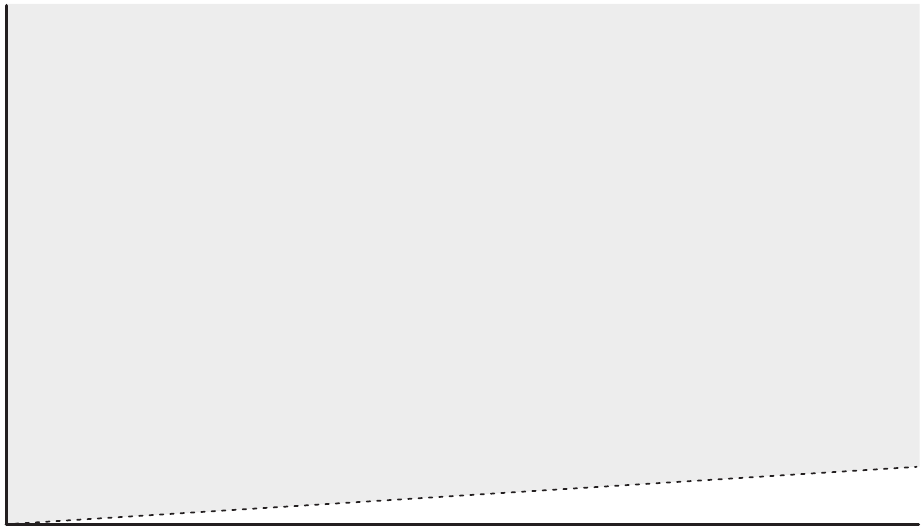}
      \put(28,30){\small $u = \tfrac{u_a}{\sqrt{t}} + O\big(\tfrac{\ln t}{t}\big)$} 
      \put(102,-1){\small $x$}
      \put(102,6){\small $t = x/N$}
      \put(0,60){\small $t$}
      \end{overpic}
     \begin{figuretext}\label{sector.pdf}
       The result of Theorem \ref{mainth} provides the long-time behavior of $u(x,t)$ in the asymptotic sector $0 \leq x \leq Nt$ (shaded).
     \end{figuretext}
     \end{center}
\end{figure}

The nonlinear steepest descent method was first introduced by Deift and Zhou in \cite{DZ1993}, where they derived the long-time asymptotics for the IVP for the modified KdV equation. The method has since then proved successful in determining asymptotic formulas for a large range of other IVPs for various integrable equations, see e.g. \cite{BV2007, DP2011, DVZ1994, DKKZ1996, HXF2015, KV1997, XF2012, XFC2013, K2008, KMM2003, BM2016, KT2009, GT2009}. By combining the ideas of \cite{DZ1993} with the unified transform formalism of \cite{F1997}, it is also possible to study asymptotics of solutions of IBVPs for nonlinear integrable PDEs \cite{BFS2004, BIK2009, BKSZ2010, BS2009, FIS2005, L2016}. Our presentation here essentially follows the approach of \cite{L2016}, where the asymptotics of the solution of the mKdV equation on the half-line was determined in the similarity and self-similar regions. 

Compared with the analysis of the IVP (see \cite{KV1997, XF2012}), the nonlinear steepest descent analysis of the half-line problem for (\ref{dnls}) presents some additional challenges. Two of these challenges are: 

\begin{enumerate}[$(a)$]
\item Whereas the RH problem relevant for the IVP for (\ref{dnls}) only has a jump across $\R \cup i\R$, the RH problem relevant for the IBVP also has jumps across the diagonal lines $e^{\frac{\pi i}{4}}\R$ and $e^{\frac{3 \pi i}{4}}\R$. The jumps across these diagonal lines can be expressed in terms of a spectral function $h(\lambda)$ whose definition involves the initial and boundary values. By introducing an analytic approximation of $h(\lambda)$ and by performing an additional deformation of the contour, we are able to compute the contribution from $h(\lambda)$ to the asymptotics.

\item The reflection coefficient which enters the RH problem for the IVP decays rapidly as the spectral parameter $k$ tends to infinity. In contrast, for the half-line problem, the analogous coefficients, in general, only decay like $1/k$ as $k \to \infty$. This absence of rapid decay has certain consequences. For example, the form of the Lax pair for (\ref{dnls}) suggests that it is beneficial to introduce a new spectral parameter $\lambda$ according to $\lambda = k^2$, see \cite{KN1978}. For the IBVP, the slow decay of the spectral functions implies that a naive introduction of $\lambda$ in the original RH problem fails. However, by first transforming the RH problem so that the off-diagonal entries of the jump matrix become $O(k^{-2})$ as $k \to \infty$, we are able to circumvent this difficulty. Another consequence of the slow decay is that we have to be more careful when introducing analytic approximations of $h(\lambda)$ and $r(\lambda)$. Indeed, these approximations have to take the asymptotic behavior of $h(\lambda)$ and $r(\lambda)$ at infinity into account. 
\end{enumerate}

\subsection{Organization of the paper}
In Section \ref{prelsec}, we briefly review how the solution of equation (\ref{dnls}) on the half-line can be expressed in terms of the solution of a RH problem. 
Our main result, Theorem \ref{mainth}, is presented in Section \ref{mainsec}. Sections \ref{specsec}-\ref{finalsec} are devoted to the proof of this result. In Section \ref{specsec}, we simplify the RH problem by introducing a new spectral parameter. In Section \ref{transsec}, we use the ideas of the nonlinear steepest descent method to transform the RH problem to a form suitable for determining the long-time asymptotics. 
In Section \ref{localsec}, we consider a local model for the RH problem near the critical point. 
Finally, by combining the above steps, we find the asymptotics of $u(x,t)$ in Section \ref{finalsec}.

\section{Preliminaries}\label{prelsec}
Equation (\ref{dnls}) is the compatibility condition of the Lax pair equations 
\begin{align}\label{mulax}
\begin{cases}
 \mu_x +ik^2[\sigma_3, \mu]  = \mathsf{U} (x,t,k)\mu, \\
\mu_t + 2ik^4[\sigma_3, \mu] = \mathsf{V} (x,t,k)\mu,
\end{cases}
\end{align}
where $k \in \C$ is the spectral parameter, $\mu(x,t,k)$ is a $2\times 2$-matrix valued eigenfunction, 
\begin{align*}
& \mathsf{U}(x,t,k) = \begin{pmatrix} -\frac{i}{2}|u|^2 & kue^{-2i\int_{(0,0)}^{(x,t)}\Delta} \\ k\bar{u}e^{2i\int_{(0,0)}^{(x,t)}\Delta} & \frac{i}{2}|u|^2\end{pmatrix}, \qquad \sigma_3 = \begin{pmatrix} 1 & 0 \\ 0 & -1 \end{pmatrix},
	\\
& \mathsf{V}(x,t,k) = \begin{pmatrix} -ik^2|u|^2-\frac{3i}{4}|u|^4-\frac{1}{2}({\bar u}_xu-{\bar u}u_x) & 
(2k^3u+iku_x+k |u|^2 u) e^{-2i\int_{(0,0)}^{(x,t)}\Delta}\\
 (2k^3{\bar u}-ik{\bar u}_x+k|u|^2 \bar{u}) e^{2i\int_{(0,0)}^{(x,t)}\Delta} &  ik^2|u|^2+\frac{3i}{4}|u|^4+\frac{1}{2}({\bar u}_xu-{\bar u}u_x) \end{pmatrix},
\end{align*}
and $\Delta$ is the closed real-valued one-form 
\begin{align}\label{Deltadef}
\Delta(x,t)=\frac{1}{2}|u|^2dx+\left( \frac{3}{4}|u|^4-\frac{i}{2}({\bar u}_xu-{\bar u}u_x)\right)dt.
\end{align}
The version (\ref{mulax}) of the Lax pair for (\ref{dnls}) is suitable for defining eigenfunctions with the appropriate asymptotics as $k \to \infty$; it was derived in \cite{L2008} by transforming the Lax pair of Kaup and Newell \cite{KN1978}.

\subsection{Riemann-Hilbert problem}
The solution of (\ref{dnls}) on the half-line can be expressed in terms of the solution of a $2\times 2$-matrix valued RH problem as follows, see \cite{L2008}. 
Let $u_0(x) = u(x,0)$ denote the initial data and let $g_0(t) = u(0,t)$ and $g_1(t) = u_x(0,t)$ denote the Dirichlet and Neumann boundary values, respectively. Suppose $u_0(x)$, $g_0(t)$, $g_1(t)$ belong to the Schwartz class $\mathcal{S}([0,\infty))$. Define four spectral functions $\{a(k), b(k), A(k), B(k)\}$ by
\begin{align}\label{abABdef}
\mathsf{X}(0,k) = \begin{pmatrix} 
\overline{a(\bar{k})} 	&	b(k)	\\
\overline{b(\bar{k})}	&	a(k)
\end{pmatrix},	\qquad 
\mathsf{T}(0,k) = \begin{pmatrix} 
\overline{A(\bar{k})} 	&	B(k)	\\
\overline{B(\bar{k})}	&	A(k)
\end{pmatrix},
\end{align}
where $\mathsf{X}(x,k)$ and $\mathsf{T}(t,k)$ are the solutions of the Volterra integral equations
\begin{align*}
 & \mathsf{X}(x,k) = I + \int_{\infty}^x e^{i k^2(x'-x)\hat{\sigma}_3} \mathsf{U}(x',0,k) \mathsf{X}(x',k) dx',
	\\
& \mathsf{T}(t,k) = I + \int_{\infty}^t e^{2i k^4(t'-t)\hat{\sigma}_3} \mathsf{V}(0,t', k) \mathsf{T}(t',k) dt',
 \end{align*}
and $e^{\hat{\sigma}_3}$ acts on a $2 \times 2$ matrix $A$ by $e^{\hat{\sigma}_3}A = e^{\sigma_3} A e^{-\sigma_3}$. 
Define the open domains $\{D_j\}_1^4$ of the complex $k$-plane by (see Figure \ref{Sigma.pdf})
\begin{align}\nonumber
D_j=\left\{ k\in \C \, |\, \arg{k^2}\in ((j-1)\pi/2,j\pi/2) \right\}, \qquad j=1,...,4.
\end{align}
Let $\Sigma = \{k \in \C \,  | \, k^4 \in \R\}$ denote the contour separating the $D_j$'s, oriented so that $D_1 \cup D_3$ lies to the left as in Figure \ref{Sigma.pdf}.
The functions $a(k)$,  $b(k)$, $A(k)$, and $B(k)$ have the following properties: 
\begin{itemize}
\item $a(k)$ and $b(k)$ are smooth and bounded on $\bar{D}_1 \cup \bar{D}_2$ and analytic in the interior of this set.

\item $a(k)\overline{a(\bar k)}-b(k)\overline{b(\bar k)}=1$ for $k \in \R \cup i\R$.

\item $a(k)=1+O(k^{-1})$ and $b(k)= b_1/k + O(k^{-2})$ uniformly as $k\rightarrow \infty$ in $\bar{D}_1 \cup \bar{D}_2$, where $b_1 = -iu_0(0)/2$.

\item $a(k)$ is an even function and $b(k)$ is an odd function of $k \in \bar{D}_1 \cup \bar{D}_2$.

\item $A(k)$ and $B(k)$ are smooth and bounded on $\bar{D}_1 \cup \bar{D}_3$ and analytic in $D_1 \cup D_3$.
\item $A(k)\overline{A(\bar k)}-B(k)\overline{B(\bar k)}=1$ for $k \in \Sigma$.

\item $A(k)=1+O(k^{-1})$ and $B(k)= B_1/k + O(k^{-2})$ uniformly as $k\rightarrow \infty$ in $\bar{D}_1 \cup \bar{D}_3$, where $B_1 = -ig_0(0)/2$.

\item $A(k)$ is an even function and $B(k)$ is an odd function of $k \in \bar{D}_1 \cup \bar{D}_3$.
\end{itemize}

\begin{figure}
\begin{center}
\begin{overpic}[width=.45\textwidth]{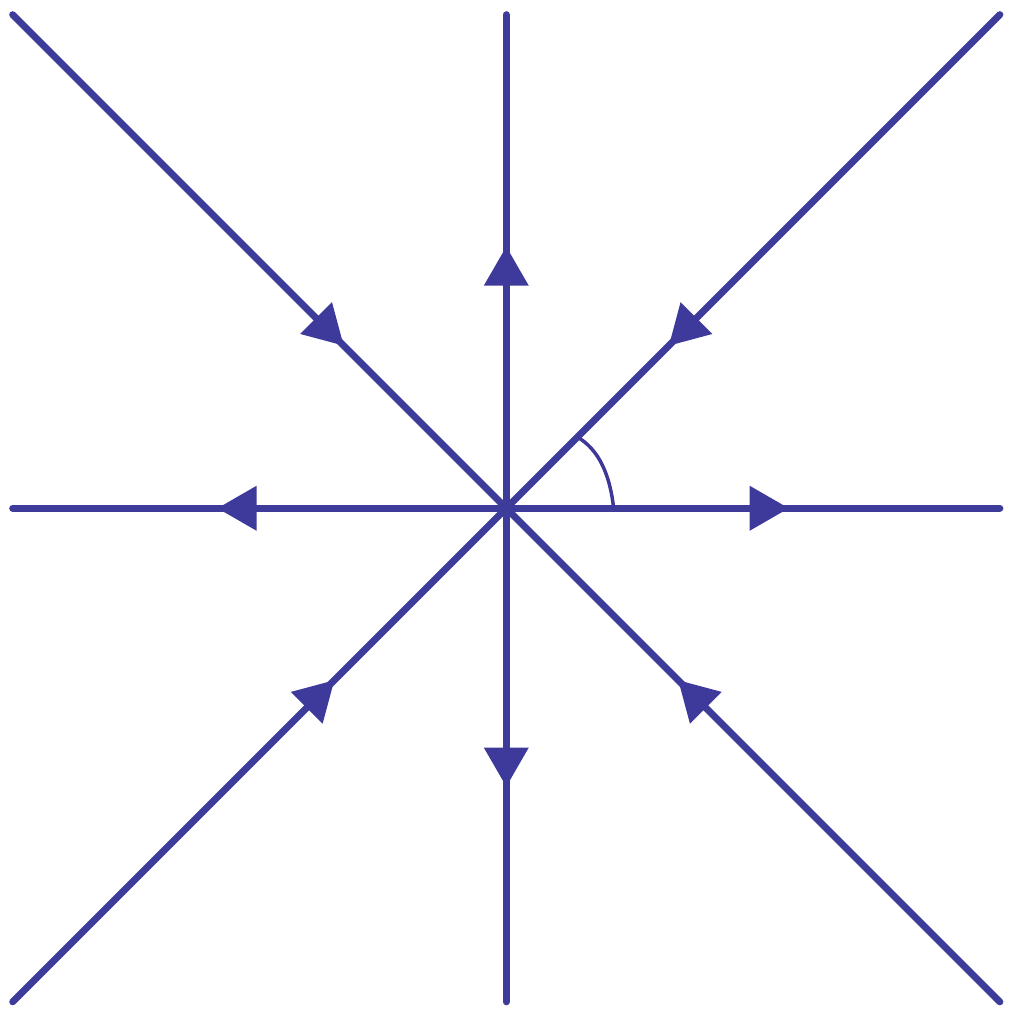}
      \put(77,59){$D_1$}
      \put(35,77){$D_3$}
      \put(57,77){$D_2$}
      \put(18,59){$D_4$}
      \put(18,37){$D_1$}
      \put(35,21){$D_2$}
      \put(57,21){$D_3$}
      \put(77,37){$D_4$}
      \put(61,54){$\pi/4$}
      \put(102,48){$\Sigma$}
      \end{overpic}
     \begin{figuretext}\label{Sigma.pdf}
       The contour $\Sigma$ and the domains $\{D_j\}_1^4$ in the complex $k$-plane.
     \end{figuretext}
     \end{center}
\end{figure}

The initial and boundary values $u_0, g_0, g_1$ cannot all be independently prescribed but must satisfy an important compatibility condition. This compatibility condition is conveniently expressed at the level of the spectral functions as the so-called global relation:
\begin{align}\label{GR}
B(k)a(k) - A(k)b(k) = 0, \qquad k \in \bar{D}_1.
\end{align}
If there exists a solution of (\ref{dnls}) in the quarter plane  $\{x \geq 0, t \geq 0\}$ with sufficient smooth and decay as $x \to \infty$, then the relation (\ref{GR}) holds \cite{L2008}. In other words, we need to assume (\ref{GR}) in order to have a well-posed problem. 
We will also assume that the functions $a(k)$ and $d(k) := a(k)\overline{A(\bar{k})} -  b(k) \overline{B(\bar{k})}$ have no zeros.
More precisely, we make the following assumptions.

\begin{assumption}\label{assumption1}
 We assume that the following conditions are satisfied:
\begin{itemize}
\item the spectral functions $a(k)$, $b(k)$, $A(k)$, $B(k)$ defined in (\ref{abABdef}) satisfy the global relation (\ref{GR}).

\item $a(k)$ and $d(k)$ have no zeros in $\bar{D}_1 \cup \bar{D}_2$ and $\bar{D}_2$, respectively.

\item the initial and boundary values $u_0(x), g_0(t)$, and $g_1(t)$ are compatible with equation (\ref{dnls})  to all orders at $x=t=0$, that is,
\begin{align*}
& g_0(0) = u_0(0), \qquad g_1(0) = u_0'(0), \qquad 
ig_0'(0) + u_0''(0) = i(|u_0|^2u_0)'(0), 
	\\
& ig_1'(0) + u_0'''(0) = i(|u_0|^2u_0)''(0), \quad \text{etc.}
\end{align*}
\end{itemize}
\end{assumption}

\begin{remark}\upshape
 The case when $a(k)$ and $d(k)$ possess a finite number of distinct zeros, can be treated by augmenting the RH problem with a set of residue conditions \cite{L2008}. The possible soliton contributions generated by these residues can easily be computed and added to the formula for the long-time asymptotics; the steps are similar to the analogous procedure for the NLS equation (see Appendix B of \cite{FIS2005}).
 \end{remark}
 
Let $f^*(k) := \overline{f(\bar{k})}$ denote the Schwartz conjugate of a function $f(k)$.
We define the jump matrix $J(x,t,k)$ by 
\begin{align}\label{Jdef}
&J(x,t,k) = \begin{cases} 
 \begin{pmatrix} 1-\frac{bb^*}{aa^*} & \frac{b}{a^*} e^{-t\Phi} \\
 - \frac{b^*}{a} e^{t\Phi}& 1 \end{pmatrix}, & k \in \R,
	\\
 \begin{pmatrix} 1 & 0 \\ -\Upsilon e^{t\Phi} & 1 \end{pmatrix}, & k \in e^{\frac{\pi i}{4}}\R,
	\\ 
	 \begin{pmatrix} 1 & -(\frac{b}{a^*}-\Upsilon^*) e^{-t\Phi} \\
 (\frac{b^*}{a}-\Upsilon)e^{t\Phi}& 1-(\frac{b^*}{a}-\Upsilon) (\frac{b}{a^*} - \Upsilon^*) \end{pmatrix}, & k \in i \R,
	\\ 
 \begin{pmatrix} 1 & \Upsilon^* e^{-t\Phi} \\ 0 & 1 \end{pmatrix}, & k \in e^{\frac{3\pi i}{4}}\R,
\end{cases}
\end{align}
where $\Phi := 4ik^4 + 2ik^2{\frac{x}{t}} $ and $\Upsilon(k)$ is defined by
$$\Upsilon(k) = \frac{ \overline{B(\bar{k})}}{a(k) d(k)}, \qquad k \in \bar{D}_2.$$ 

We will need to formulate RH problems in the $L^2$ sense. To this end, it is convenient to introduce the generalized Smirnoff class $\dot{E}^2(\C\setminus \Sigma)$. The space $\dot{E}^2(\C\setminus \Sigma)$ consists of all functions $f(k)$ analytic in $\C\setminus \Sigma$ with the property that for each component $D$ of $\C\setminus \Sigma$ there exist curves $\{C_n\}_1^\infty$ in $D$ such that the $C_n$ eventually surround each compact subset of $D$ and $\sup_{n \geq 1} \|f\|_{L^2(C_n)} < \infty$.
For $f \in \dot{E}^2(\C\setminus \Sigma)$, we let $f_+$ and $f_-$ denote the nontangential boundary values of $f(k)$ as $k$ approaches $\Sigma$ from the left and right, respectively (these boundary values exist almost everywhere on $\Sigma$). We write $E^\infty(D)$ for the set of bounded analytic functions in $D$.

The arguments of \cite{L2008} show that if Assumption \ref{assumption1} is satisfied, then the RH problem
\begin{align}\label{RHM}
\begin{cases}
M(x, t, \cdot) \in I + \dot{E}^2(\C \setminus \Sigma),\\
M_+(x,t,k) = M_-(x, t, k) J(x, t, k) \quad \text{for a.e.} \ k \in \Sigma,
\end{cases}
\end{align}
has a unique solution for each $(x,t) \in [0,\infty) \times [0, \infty)$ and the nontangential limit\footnote{The notation $\displaystyle{\ntlim_{k \to \infty}}$ indicates the limit as $k \in \C$ goes to $\infty$ nontangentially with respect to $\Sigma$.}
\begin{align}\label{tildeudef}
  \tilde{u}(x,t) := \ntlim_{k \to \infty} \left(k M(x,t, k)\right)_{12}
\end{align}
exists for each $(x,t) \in [0,\infty) \times [0, \infty)$. 
Moreover, the function $u(x,t)$ defined by 
\begin{align}\label{urecover}
& u(x,t) = 2i\tilde{u}(x,t)e^{2i\int^{(x,t)}_{(0, 0)} \Delta}, \qquad x \geq 0, \ t\geq 0,
\end{align}
with
\begin{align}
& \Delta = 2 |\tilde{u}|^2 dx - \left(4|\tilde{u}|^2 + 2i\left(\bar{\tilde{u}}_x\tilde{u} - \tilde{u}_x\bar{\tilde{u}}\right)\right)dt,
\end{align}
is a smooth solution of the DNLS equation (\ref{dnls}) in the quarter plane $\{x \geq 0, t \geq0\}$ such that 
\begin{align}\label{u0g0g1}
\begin{cases} \text{$u(x,0) = u_0(x)$ for $x \geq 0$},
	\\ 
 \text{$u(0,t) = g_0(t)$ and $u_x(0,t) = g_1(t)$ for $t \geq 0$.} 
 \end{cases}
\end{align}

\section{Main result}\label{mainsec}
Our main result gives an explicit formula for the long-time asymptotics of the solution $u(x,t)$ of the DNLS equation on the half-line under the assumption that the initial and boundary values lie in the Schwartz class. 
The formula is valid in the sector displayed in Figure \ref{sector.pdf}.

\begin{theorem}[Asymptotics for the DNLS equation on the half-line]\label{mainth}
  Let $u_0(x)$, $g_0(t)$, $g_1(t)$ be functions in the Schwartz class $\mathcal{S}([0,\infty))$. Define the spectral functions $a(k)$, $b(k)$, $A(k)$, $B(k)$ via (\ref{abABdef}) and suppose Assumption \ref{assumption1} holds. 
  Let $u(x,t)$ be the half-line solution of (\ref{dnls}) defined in (\ref{urecover}) with initial and boundary values 
  $$u(x,0) = u_0(x), \qquad u(0,t) = g_0(t), \qquad u_x(0,t) = g_1(t).$$

Then, for any constant $N > 0$, $u(x,t)$ satisfies the asymptotic formula
\begin{align}\label{uasymptotics}
u(x,t) = \frac{u_a(x,t)}{\sqrt{t}} + O\bigg(\frac{\ln t}{t}\bigg), \qquad t \to \infty, \ 0 \leq x \leq Nt,
\end{align}
where the error term is uniform with respect to $x$ in the given range, and the leading-order coefficient $u_a(x,t)$ is defined as follows: 
Let $\zeta = x/t$ and define $\lambda_0 := \lambda_0(\zeta) \leq 0$ by
$$\lambda_0 = -\frac{\zeta}{4} = -\frac{x}{4t}.$$
Define the spectral function $r(\lambda)$ for $\lambda=k^2 \leq 0$ in terms of the functions $a(k)$, $b(k)$, $A(k)$, $B(k)$  by equation (\ref{hrdef}).
Define the function $\nu(\zeta)$ for $\zeta \leq 0$ by
$$\nu(\zeta) = \frac{1}{2\pi} \ln(1 - \lambda_0 |r(\lambda_0)|^2) \geq 0.$$

Then 
\begin{align*}
u_a(x,t) = &\; \sqrt{\frac{\nu(\zeta)}{2|\lambda_0|}}e^{i\alpha(\zeta,t)}, 
\end{align*}
where the phase $\alpha(\zeta, t)$ is given by
\begin{align*}
\alpha(\zeta,t)= &\;
\frac{\pi}{4} - \arg r(\lambda_0) - \arg \Gamma(i\nu(\zeta)) + \nu(\zeta) \ln(8t) 
	\\
& + \frac{1}{\pi} \int_{-\infty}^{\lambda_0} \ln (\lambda_0-s) d\ln(1 - s|r(s)|^2) + 4t\lambda_0^2 
 	\\
&+ \frac{1}{\pi}\int_0^{|\lambda_0|} \frac{\ln(1+s|r(-s)|^2)}{s}ds 
+ 2\int_0^t\bigg(\frac{3}{4}|g_0|^4 - \frac{i}{2}(\bar{g}_1g_0 - g_1\bar{g}_0)\bigg)dt'.
\end{align*}
\end{theorem}

\begin{remark}[Effect of the boundary]\upshape\label{effectremark}
Theorem \ref{mainth} can be used to understand the influence of the boundary on the asymptotic behavior of the solution $u(x,t)$. 
The asymptotic formula (\ref{uasymptotics}), which is valid uniformly for $0 \leq x \leq Nt$, has the leading order term $u_a(x,t)t^{-1/2}$ where $u_a(x,t)$ is defined in terms of the spectral function $r(\lambda)$. 
If $u^L(x,t)$ denotes the solution of the initial value problem for equation (\ref{dnls}) on the line (i.e., in the absence of a boundary), then the asymptotics of  $u^L(x,t)$ is given by a formula similar to equation (\ref{uasymptotics}) except that $r(\lambda)$ is replaced by another spectral function  $r^L(\lambda)$ whose definition involves only the initial data. In general, the solution $u^L(x,t)$ decays like $t^{-1/2}$ in a sector around $x = 0$. On the other hand, in Theorem \ref{mainth}, we have assumed that a boundary is present at $x = 0$ and that $u(0,t)$ has rapid decay as $t \to \infty$. This rapid decay can only be consistent with formula (\ref{uasymptotics}) provided that $\zeta^{-1} \nu(\zeta) \to 0$ as $\zeta \to 0$. 
In fact, it turns out that the global relation (\ref{GR}) forces $\nu(\zeta)$ to vanish to all orders at $\zeta = 0$, see Remark \ref{r0remark}. This shows that not only is $u_a(0,t) = 0$, but in fact the leading-order coefficient $u_a(x,t)$ in (\ref{uasymptotics}) satisfies $u_a(x,t) = O(|x/t|^n)$ for every $n \geq 1$ as $x/t \to 0$. This illustrates the effect of the boundary on the long-time behavior of the solution.
\end{remark}

The remainder of the paper is devoted to the proof of Theorem \ref{mainth}. We henceforth assume $u_0(x)$, $g_0(t)$, $g_1(t)$ are given functions in $\mathcal{S}([0,\infty))$ with associated spectral functions $a(k)$, $b(k)$, $A(k)$, $B(k)$ such that Assumption \ref{assumption1} holds. We also fix $N >0$ and let $\mathcal{I}$ denote the interval $\mathcal{I} = (0,N]$.
We let $\{\mathcal{D}_j\}_1^4$ denote the four open quadrants of the complex $\lambda$-plane (see Figure \ref{Gamma.pdf})
\begin{align}\nonumber
\mathcal{D}_j=\left\{\lambda \in \C \, |\, \arg\lambda \in ((j-1)\pi/2,j\pi/2) \right\}, \qquad j=1,...,4,
\end{align}
and let $\Gamma = \R \cup i\R$ denote the contour separating the ${\mathcal D}_j$ oriented as in Figure \ref{Gamma.pdf}.

\begin{figure}
\begin{center}
\begin{overpic}[width=.45\textwidth]{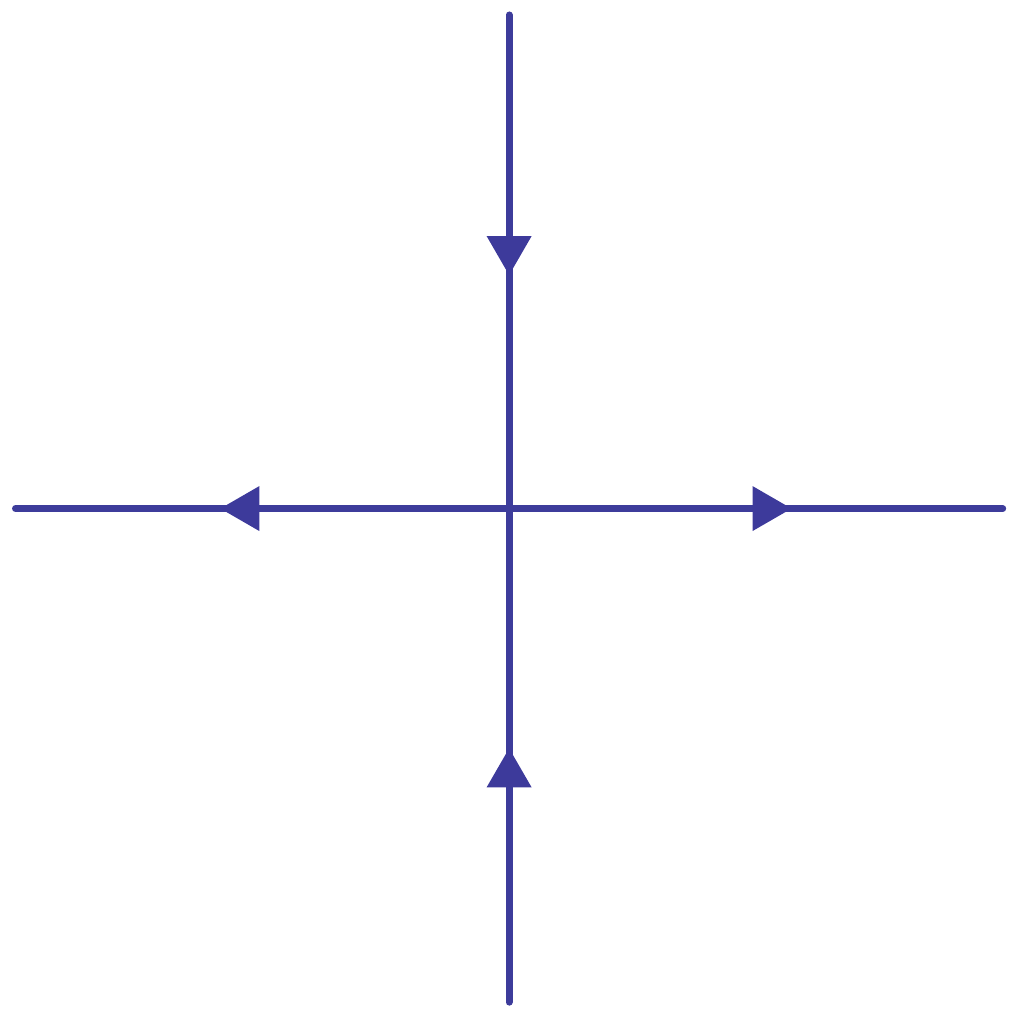}
      \put(74,74){$\mathcal{D}_1$}
      \put(20,74){$\mathcal{D}_2$}
      \put(20,22){$\mathcal{D}_3$}
      \put(74,22){$\mathcal{D}_4$}
      \put(102,48){$\Gamma$}
      \end{overpic}
     \begin{figuretext}\label{Gamma.pdf}
       The contour $\Gamma = \R \cup i\R$ and the quadrants $\{\mathcal{D}_j\}_1^4$ of the complex $\lambda$-plane.
     \end{figuretext}
     \end{center}
\end{figure}

\section{A new spectral parameter}\label{specsec}
The proof of Theorem \ref{mainth} relies on a nonlinear steepest descent analysis of the RH problem (\ref{RHM}). Although it is possible to perform this analysis directly in the complex $k$-plane, the symmetries
\begin{align}\label{abABsymm}
a(k) = a(-k), \quad b(k) = -b(-k), \quad A(k) = A(-k), \quad B(k) = -B(-k),
\end{align}
suggest that it is convenient to introduce a new spectral variable $\lambda$ by $\lambda = k^2$. This change of spectral parameter appeared already in \cite{KN1978} and was employed in \cite{XF2012} for the analysis of the IVP for (\ref{dnls}). One advantage of working with $\lambda$ is that we only have to analyze a single critical point of $\Phi (\zeta,\lambda)$ in order to find the long-time asymptotics. 

\subsection{The solution $M^{(1)}$}
One difference between the IVP and the IBVP is that the reflection coefficient which enters the RH problem for the IVP decays rapidly as $k \to \infty$. In contrast, for the half-line problem, the analogous coefficients typically only decay like $1/k$ as $k \to \infty$. Therefore, before we introduce the new parameter $\lambda$, we transform the RH problem (\ref{RHM}) by defining the sectionally analytic function $M^{(1)}$ by
$$M^{(1)}(x,t,k) = M(x,t,k)G(x,t,k),$$
where
$$G(x,t,k) = \begin{cases}
\begin{pmatrix} 1 & 0 \\ \frac{k\bar{b}_1}{k^2 + i} e^{t\Phi} & 1 \end{pmatrix}, & k \in D_1, 
	\\
\begin{pmatrix} 1 & \frac{kb_1}{k^2 - i} e^{-t\Phi} \\ 0 & 1 \end{pmatrix}, & k \in D_4, 
 	\\
I, & \text{elsewhere}.
\end{cases}
$$
The factor $\frac{k\bar{b}_1}{k^2 + i}$ is an odd function of $k$ which is analytic in $D_1$ (the poles lie at $k = \pm e^{\frac{3\pi i}{4}}$) and such that
$$\frac{k\bar{b}_1}{k^2 + i} = \frac{\bar{b}_1}{k} + O\bigg(\frac{1}{k^2}\bigg), \qquad k \to \infty.$$
Similarly, $\frac{k b_1}{k^2 - i}$ is an odd analytic function of $k$ in $D_4$ such that
$$\frac{kb_1}{k^2 - i} = \frac{b_1}{k} + O\bigg(\frac{1}{k^2}\bigg), \qquad k \to \infty.$$
In particular,
$$G(x,t,\cdot) \in I + (\dot{E}^2 \cap E^\infty)(\C \setminus \Sigma).$$
On the other hand, the compatibility condition $u_0(0) = g_0(0)$ implies 
$$b_1 = B_1 = -\frac{iu_0(0)}{2}.$$
It follows that $M(x,t,k)$ satisfies the RH problem (\ref{RHM}) if and only if $M^{(1)}(x,t,k)$ satisfies the RH problem
\begin{align}\label{RHM1}
\begin{cases}
M^{(1)}(x, t, \cdot) \in I + \dot{E}^2(\C \setminus \Sigma),\\
M_+^{(1)}(x,t,k) = M_-^{(1)}(x, t, k) J^{(1)}(x, t, k) \quad \text{for a.e.} \ k \in \Sigma,
\end{cases}
\end{align}
where $J^{(1)}$ is given by
\begin{align}\label{J1def}
&J^{(1)}(x,t,k) = \begin{cases} 
 \begin{pmatrix} 1- (\frac{b}{a^*} - \frac{k b_1}{k^2 - i})(\frac{b^*}{a} - \frac{k\bar{b}_1}{k^2 + i}) & (\frac{b}{a^*} - \frac{k b_1}{k^2 - i}) e^{-t\Phi} \\
 - (\frac{b^*}{a} - \frac{k\bar{b}_1}{k^2 + i})e^{t\Phi}& 1 \end{pmatrix}, & k \in \R,
	\\
 \begin{pmatrix} 1 & 0 \\ -(\Upsilon - \frac{k\bar{b}_1}{k^2 + i}) e^{t\Phi} & 1 \end{pmatrix}, & k \in e^{\frac{\pi i}{4}}\R,
	\\ 
 \begin{pmatrix} 1 & -(\frac{b}{a^*}-\Upsilon^*) e^{-t\Phi} \\
 (\frac{b^*}{a}-\Upsilon)e^{t\Phi}& 1-(\frac{b^*}{a}-\Upsilon) (\frac{b}{a^*} - \Upsilon^*) \end{pmatrix}, & k \in i \R,
	\\ 
 \begin{pmatrix} 1 & (\Upsilon^* - \frac{k b_1}{k^2 - i}) e^{-t\Phi} \\ 0 & 1 \end{pmatrix}, & k \in e^{\frac{3\pi i}{4}}\R,
\end{cases}
\end{align}
The jump matrix $J^{(1)}$ has the advantage that the off-diagonal entries are $O(k^{-2})$ as $k \to \infty$.

\subsection{The new parameter}
The relations (\ref{abABsymm}) imply that the solution $M^{(1)}(x,t,k)$ of (\ref{RHM1}) obeys the symmetry 
$$M^{(1)}(x,t,k) = \sigma_3 M^{(1)}(x,t,-k) \sigma_3, \qquad k \in \C \setminus \Sigma.$$ 
Hence, letting $\lambda = k^2$, we can define $m(x,t,\lambda)$ by the equation
\begin{align}\label{mdef}
m(x,t,\lambda) = \begin{pmatrix} 1 & 0 \\ -\overline{\tilde{u}(x,t)} & 1 \end{pmatrix} k^{-\frac{\hat{\sigma}_3}{2}} M^{(1)}(x,t,k), \qquad \lambda \in \C \setminus \Gamma.
\end{align}
The factor $k^{-\frac{\hat{\sigma}_3}{2}}$ is included in (\ref{mdef}) in order to make the right-hand side an even function of $k$; the matrix involving $\tilde{u}$ is included to ensure that $m \to I$  as $\lambda \to \infty$. 

Let $\R_+ = [0,\infty)$ and $\R_- = (-\infty,0]$ denote the positive and negative halves of the real axis.
The relations (\ref{abABsymm}) show that we may define functions  $h(\lambda)$, $r_1(\lambda)$, and $r(\lambda)$ by
\begin{subequations}\label{hrdef}
\begin{align}
& h(\lambda) = -\frac{\Upsilon(k)}{k} + \frac{\bar{b}_1}{k^2 + i}, \qquad \lambda \in \bar{\mathcal{D}}_2,
	\\ 
& r_1(\lambda) = \frac{\overline{b({\bar k})}}{ka(k)} - \frac{\bar{b}_1}{k^2 + i}, \qquad \lambda \in \R,
	\\	
& r(\lambda) = r_1(\lambda) + h(\lambda), \qquad \lambda \in \R_-.
\end{align}
\end{subequations}
These functions possess the following properties:
\begin{itemize}
\item $h(\lambda)$ is smooth and bounded on $\bar{\mathcal{D}}_2$ and analytic in $\mathcal{D}_2$.
\item $r_1(\lambda)$ is smooth and bounded on $\R$.
\item $r(\lambda)$ is smooth and bounded on $(-\infty,0]$.
\item There exist complex constants $\{h_j\}_{j=2}^{\infty}$ and $\{r_j\}_{2=1}^{\infty}$ such that, for any $N \geq 1$, 
\begin{subequations}
\begin{align}\label{hexpansion}
& h(\lambda) = \sum_{j=2}^N \frac{h_j}{\lambda^j} + O\Big(\frac{1}{\lambda^{N+1}}\Big), \qquad \lambda \to \infty, \ \lambda \in \bar{\mathcal{D}}_2,
	\\
& r(\lambda) = \sum_{j=2}^N \frac{r_j}{\lambda^j}+O\Big(\frac{1}{\lambda^{N+1}}\Big), \qquad |\lambda| \to \infty, \ \lambda \in \R_-,
\end{align}
\end{subequations}
and these expansions can be differentiated termwise any number of times. 
\end{itemize}

It follows that the RH problem (\ref{RHM1}) for $M^{(1)}(x,t,k)$ can be rewritten in terms of $m(x,t,\lambda)$ as
\begin{align}\label{RHm}
\begin{cases}
m(x, t, \cdot) \in I + \dot{E}^2(\C \setminus \Gamma),\\
m_+(x,t,\lambda) = m_-(x, t, \lambda) v(x, t, \lambda) \quad \text{for a.e.} \ \lambda \in \Gamma,
\end{cases}
\end{align}
where the jump matrix $v = k^{-\frac{\hat{\sigma}_3}{2}} J^{(1)}$ is given by
\begin{align}\label{vdef}
&v(x,t,\lambda) = \begin{cases} 
 \begin{pmatrix} 1- \lambda r_1 r_1^* & r_1^* e^{-t\Phi} \\
 - \lambda r_1 e^{t\Phi}& 1 \end{pmatrix}, & \lambda \in \R_+,
	\\ 
 \begin{pmatrix} 1 & 0 \\ \lambda h e^{t\Phi} & 1 \end{pmatrix}, & \lambda \in i\R_+,
	\\
	 \begin{pmatrix} 1 & -r^* e^{-t\Phi} \\
 \lambda r e^{t\Phi}& 1- \lambda rr^* \end{pmatrix}, & \lambda \in \R_-,
	\\ 
 \begin{pmatrix} 1 & -h^* e^{-t\Phi} \\ 0 & 1 \end{pmatrix}, & \lambda \in i\R_-,
\end{cases}
\end{align}
with $\Phi(\zeta, \lambda)$ defined by
$$\Phi(\zeta, \lambda) = 4i\lambda^2 + 2i\zeta \lambda.$$
In terms of $m$, equation (\ref{tildeudef}) can be expressed as
\begin{align}\label{tildeurecover}
\tilde{u}(x,t) = \ntlim_{\lambda \to \infty} \left(\lambda m(x,t, \lambda)\right)_{12}.
\end{align}

\begin{remark}[Behavior of $r(\lambda)$ as $\lambda \to 0$]\label{r0remark}\upshape
We claimed in Remark \ref{effectremark} that the global relation implies that $r(\lambda)$ vanishes to all orders at $\lambda = 0$, i.e., that $r^{(j)}(0) = 0$ for all $j = 0, 1, \dots$. We are now in a position to prove this.
Using the fact that $aa^* - bb^*=1$ for $k \in \R \cup i\R$, simplification of the definition (\ref{hrdef}) of $r(\lambda)$ shows that
$$kr(k^2) = \frac{\overline{c(\bar{k})}}{d(k)},\qquad k \in i \R,$$
where $c(k) = b(k)A(k) - B(k)a(k)$. 
The functions $a(k)$ and $b(k)$ are analytic for $\im k > 0$ and, for each $j \geq 0$, the derivatives $a^{(j)}(k), b^{(j)}(k)$ have continuous extensions to $\im k \geq 0$. Similarly, the functions $A(k)$ and $B(k)$ are analytic for $k \in D_1 \cup D_3$ and, for each $j \geq 0$, the derivatives $A^{(j)}(k), B^{(j)}(k)$ have continuous extensions to $\bar{D}_1 \cup \bar{D}_3$. (These properties follow from an analysis of the Volterra equations satisfied by $\mathsf{X}$ and $\mathsf{T}$ together with the assumption that the initial and boundary values belong to the Schwartz class; see \cite{LNonlinearFourier} for detailed proofs of such properties in the similar case of the modified KdV equation.) We infer that the functions $c(k)$ and $d(k)$ have smooth extensions to $\R \cup i\R$, where smooth means that the restrictions $c|_\R$, $c|_{i\R}$, $d|_\R$, $d|_{i\R}$ are $C^\infty$ and the derivatives along $\R$ and $i\R$ are consistent to all orders at $k = 0$, i.e., for $s$ real,
$$\frac{d^n}{ds^n}\bigg|_{s=0}c(is) = i^n\frac{d^n}{ds^n}\bigg|_{s=0}c(s), \quad
\frac{d^n}{ds^n}\bigg|_{s=0}d(is) = i^n\frac{d^n}{ds^n}\bigg|_{s=0}d(s), \qquad \ n = 1,2,\dots.$$
In terms of the functions $\mathsf{X}(0,k)$ and $\mathsf{T}(0,k)$ defined in (\ref{abABdef}), we can write
$$\mathsf{T}(0,k)^{-1}\mathsf{X}(0,k) = \begin{pmatrix} 
\overline{d(\bar{k})} 	&	c(k)	\\
\overline{c(\bar{k})}	&	d(k)
\end{pmatrix}.$$
Since $\det \mathsf{X}(0,k) = \det \mathsf{T}(0,k) = 1$ for $k \in \R$, we conclude that
$$|d(k)|^2 = 1 + |c(k)|^2,  \qquad k \in \R.$$
In particular, $d(k)$ is nonzero on $\R$. By Assumption \ref{assumption1}, $d(k)$ is also nonzero on $i\R$.
It follows that $\overline{c(\bar{k})}/d(k)$ is a smooth function $\R \cup i\R \to \C$.
However, the global relation (\ref{GR}) implies that $c(k)$ vanishes identically for $k \in \R$. 
This shows that $k \mapsto kr(k^2) = \overline{c(\bar{k})}/d(k)$ vanishes to all orders at $k = 0$; hence $r(\lambda)$ vanishes to all orders at $\lambda = 0$.
\end{remark}

\section{Transformations of the RH problem}\label{transsec}
By performing a number of transformations, we can bring the RH problem (\ref{RHm}) to a form suitable for determining the long-time asymptotics. More precisely, starting with $m$, we will define functions $m^{(j)}(x,t,\lambda)$, $j = 1, 2,3$, such that the RH problem satisfied by $m^{(j)}$ is equivalent to the original RH problem (\ref{RHm}). The RH problem for $m^{(j)}$ has the form
\begin{align}\label{RHmj}
\begin{cases}
m^{(j)}(x, t, \cdot) \in I + \dot{E}^2(\C \setminus \Gamma^{(j)}),\\
m^{(j)}_+(x,t,\lambda) = m^{(j)}_-(x, t, \lambda) v^{(j)}(x, t, \lambda) \quad \text{for a.e.} \ \lambda \in \Gamma^{(j)},
\end{cases}
\end{align}
where the contours $\Gamma^{(j)}$ and jump matrices $v^{(j)}$ are specified below. 
The jump matrix $v^{(3)}$ obtained after the third transformation has the property that it approaches the identity matrix as $t \to \infty$ everywhere on the contour except near a single critical point $\lambda_0$. This means that we can find the long-time asymptotics of $m^{(3)}$ by computing the contribution from a small cross centered at $\lambda_0$. 

\subsection{First transformation}
\begin{figure}
\begin{center}
\begin{overpic}[width=.55\textwidth]{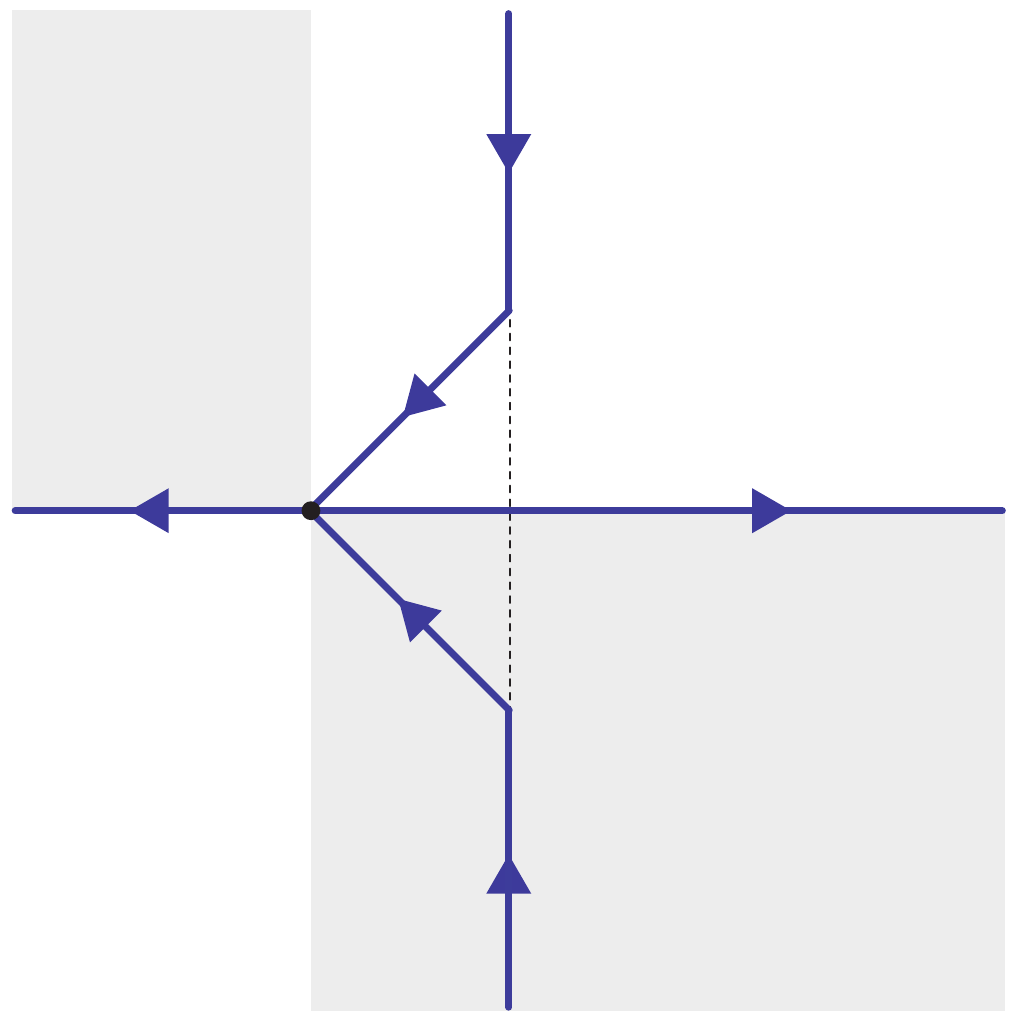}
      \put(102,48){$\Gamma^{(1)}$}
      \put(68,78){\small $\re \Phi < 0$}
      \put(68,20){\small $\re \Phi > 0$}
      \put(42,54){\small $U_1$}
      \put(42,43){\small $U_2$}
      \put(75,52){\small $1$}
      \put(37,60){\small $2$}
      \put(45,84){\small $2$}
      \put(13,52){\small $3$}
      \put(37,36){\small $4$}
      \put(45,12){\small $4$}
      \put(27,45){\small $\lambda_0$}
    \end{overpic}
     \begin{figuretext}\label{Gamma1.pdf}
       The contour $\Gamma^{(1)}$ in the complex $\lambda$-plane. The region where $\re \Phi > 0$ is shaded. 
     \end{figuretext}
     \end{center}
\end{figure}
By solving the equation $\frac{\partial \Phi}{\partial \lambda} = 0$, we see that the jump matrix $v(x,t,\lambda)$ defined in (\ref{vdef}) admits a single critical point as $t \to \infty$ located at 
$$\lambda_0 := \lambda_0(\zeta) = -\frac{\zeta}{4} < 0.$$ 
The purpose of the first transformation is to deform the vertical part of $\Gamma$ so that it passes through the critical point $\lambda_0$. Letting $U_1$ and $U_2$ denote the triangular domains shown in Figure \ref{Gamma1.pdf}, we define $m^{(1)}(x,t,\lambda)$ by
$$m^{(1)}(x,t,\lambda) = m(x,t,\lambda) \times \begin{cases} 
 \begin{pmatrix} 1 & 0 \\ \lambda h(\lambda) e^{t\Phi(\zeta, \lambda)} & 1 \end{pmatrix}, & \lambda \in U_1, 
 	\\
 \begin{pmatrix} 1 & h^*(\lambda) e^{-t\Phi(\zeta, \lambda)} \\ 0 & 1 \end{pmatrix}, & \lambda \in U_2, 
 	\\
I, & \text{elsewhere}.
\end{cases}$$
Since $h e^{t\Phi}$ and $\lambda h^* e^{-t\Phi}$ are bounded and analytic functions of $\lambda \in U_1$ and $\lambda \in U_2$, respectively, we infer that $m$ satisfies the RH problem (\ref{RHm}) if and only if $m^{(1)}$ 
satisfies the RH problem (\ref{RHmj}) with $j = 1$, where the contour $\Gamma^{(1)}$  is as in Figure \ref{Gamma1.pdf} and the jump matrix $v^{(1)}$ is given by
\begin{align*}\nonumber
&v_1^{(1)} =  \begin{pmatrix} 1-\lambda r_1 r_1^* & r_1^* e^{-t\Phi} \\ - \lambda r_1 e^{t\Phi}& 1 \end{pmatrix}, 
\qquad
v_2^{(1)} = \begin{pmatrix} 1 & 0 \\ \lambda h e^{t\Phi} & 1 \end{pmatrix}, 
	\\ 
& v_3^{(1)} = \begin{pmatrix} 1 & -r^* e^{-t\Phi} \\ \lambda r e^{t\Phi}& 1- \lambda rr^* \end{pmatrix}, 
\qquad
v_4^{(1)} = \begin{pmatrix} 1 & -h^* e^{-t\Phi} \\ 0 & 1 \end{pmatrix}.
\end{align*}
Here $v_j^{(1)}$ denotes the restriction of $v^{(1)}$ to the contour labeled by $j$ in Figure \ref{Gamma1.pdf}.

\subsection{Second transformation} 
The jump matrix $v^{(1)}$ has the wrong factorization for $\lambda \in (-\infty, \lambda_0)$. Hence we introduce $m^{(2)}$ by
$$m^{(2)}(x,t,\lambda) = m^{(1)}(x,t,\lambda) \delta(\zeta, \lambda)^{-\sigma_3},$$
where the complex-valued function $\delta(\zeta, \lambda)$ is defined by
$$\delta(\zeta, \lambda) = e^{\frac{1}{2\pi i}\int_{-\infty}^{\lambda_0} \frac{\ln(1 - s |r(s)|^2)}{s-\lambda} ds}, \qquad \lambda \in \C \setminus (-\infty, \lambda_0].$$

\begin{lemma}\label{deltalemma}
For each $\zeta \in \mathcal{I}$, the function $\delta(\zeta, \lambda)$ has the following properties:
\begin{enumerate}[$(a)$]

\item $\delta(\zeta, \lambda)$ and $\delta(\zeta, \lambda)^{-1}$ are bounded and analytic functions of $\lambda \in \C \setminus (-\infty, \lambda_0]$ with continuous boundary values on $(-\infty, \lambda_0)$.

\item $\delta$ obeys the symmetry 
$$\delta(\zeta, \lambda) = \overline{\delta(\zeta, \bar{\lambda})}^{-1}, \qquad \lambda \in \C \setminus (-\infty, \lambda_0].$$

\item Across the subcontour $(-\infty, \lambda_0)$ of $\Gamma^{(1)}$ oriented as in Figure \ref{Gamma1.pdf}, 
$\delta$ satisfies the jump condition 
$$\delta_+(\zeta, \lambda) = \frac{\delta_-(\zeta, \lambda)}{1 - \lambda |r(\lambda)|^2}, \qquad \lambda \in (-\infty, \lambda_0).$$

\item As $\lambda$  goes to infinity,  
\begin{align}\label{deltaasymptotics}
\delta(\zeta, \lambda) = 1 + O(\lambda^{-1}), \qquad \lambda \to \infty,
\end{align}
uniformly with respect to $\arg \lambda \in [0,2\pi]$. 
\end{enumerate}
\end{lemma}
\begin{proof}
Note that $\ln(1- s|r(s)|^2)$ is a smooth function of $s \in (-\infty, \lambda_0]$ such that  $\ln(1 - s|r(s)|^2) = O(|s|^{-3})$ as $s \to -\infty$. The proof is now standard. 
\end{proof}

Lemma \ref{deltalemma} implies that
$$\delta(\zeta,\cdot)^{\sigma_3} \in I + (\dot{E}^2 \cap E^\infty)(\C \setminus (-\infty, \lambda_0]), \qquad \zeta \in \mathcal{I}.$$
Hence $m$ satisfies the RH problem (\ref{RHm}) if and only if $m^{(2)}$ 
satisfies the RH problem (\ref{RHmj}) with $j = 2$, where $\Gamma^{(2)} = \Gamma^{(1)}$ and the jump matrix $v^{(2)}$ is given by $v^{(2)} =  \delta_-^{\sigma_3} v^{(1)}  \delta_+^{-\sigma_3}$, that is,
\begin{align*}\nonumber
&v_1^{(2)} =  \begin{pmatrix} 1-\lambda r_1 r_1^* & \delta^2 r_1^* e^{-t\Phi} \\ - \delta^{-2} \lambda r_1 e^{t\Phi}& 1 \end{pmatrix}, 
\qquad
v_2^{(2)} = \begin{pmatrix} 1 & 0 \\ \delta^{-2} \lambda h e^{t\Phi} & 1 \end{pmatrix}, 
	\\ 
& v_3^{(2)} = \begin{pmatrix} 1- \lambda rr^* & -\delta_-^2 \frac{r^*}{1-\lambda rr^*} e^{-t\Phi} \\  \delta_+^{-2} \lambda \frac{r}{1-\lambda rr^*} e^{t\Phi}& 1 \end{pmatrix}, 
\qquad
v_4^{(2)} = \begin{pmatrix} 1 & - \delta^2 h^* e^{-t\Phi} \\ 0 & 1 \end{pmatrix}, 
\end{align*}
If we define $r_2(\lambda)$ by 
\begin{align}\label{r2def}
  r_2(\lambda) = \frac{r^*(\lambda)}{1-\lambda r(\lambda)r^*(\lambda)}, 
\end{align}
we can write the jumps across the real axis as 
\begin{align}\nonumber
&v_1^{(2)} =  \begin{pmatrix} 1 & \delta^2 r_1^* e^{-t\Phi} \\ 0 & 1 \end{pmatrix}
\begin{pmatrix} 1 & 0 \\ - \delta^{-2} \lambda r_1 e^{t\Phi}& 1 \end{pmatrix}, 
	\\
& v_3^{(2)} = \begin{pmatrix} 1 & -\delta_-^2 r_2 e^{-t\Phi} \\ 0 & 1 \end{pmatrix}
 \begin{pmatrix} 1 & 0 \\   \delta_+^{-2} \lambda r_2^* e^{t\Phi}& 1 \end{pmatrix}.
\end{align}

\subsection{Third transformation}
The purpose of the third transformation is to deform the contour so that the jump matrix involves the exponential factor $e^{-t\Phi}$ on the parts of the contour where $\re \Phi$ is positive and the factor $e^{t\Phi}$ on the parts where $\re \Phi$ is negative. Since the spectral functions have limited domains of analyticity, we follow the idea of \cite{DZ1993} and decompose each of the functions $h, r_1, r_2$ into an analytic part and a small remainder. The analytic part of the jump matrix will be deformed, whereas the small remainder will be left on the original contour.

\begin{lemma}[Analytic approximation of $h$]\label{decompositionlemma}
There exists a decomposition
\begin{align*}
h(\lambda) = h_{a}(t, \lambda) + h_{r}(t, \lambda), \qquad t > 0, \ \lambda \in i \R_+,  
\end{align*}
where the functions $h_{a}$ and $h_{r}$ have the following properties:
\begin{enumerate}[$(a)$]
\item For each $t > 0$, $h_{a}(t, \lambda)$ is defined and continuous for $\lambda \in \bar{{\mathcal D}}_1$ and analytic for $\lambda \in \mathcal{D}_1$.

\item The function $h_a$ satisfies
\begin{align}\label{haestimate}
\begin{cases} 
|h_{a}(t, \lambda) - h(0)| \leq C|\lambda| e^{\frac{t}{4} |\re \Phi(\zeta,\lambda )|},
	\\
|h_{a}(t, \lambda)| \leq 
\frac{C}{1 + |\lambda|^2} e^{\frac{t}{4} |\re \Phi(\zeta,\lambda )|},	
\end{cases}  \qquad \lambda \in \bar{\mathcal D}_1, \ \zeta \in \mathcal{I}, \ t > 0,
\end{align}
where the constant $C$ is independent of $\zeta, t, \lambda$.

\item The $L^1, L^2$, and $L^\infty$ norms of the function $h_{r}(t, \cdot)$ on $i \R_+$ are $O(t^{-3/2})$ as $t \to \infty$.
\end{enumerate}
\end{lemma}
\begin{proof}
See Appendix \ref{appA}.
\end{proof}

Let $V_j := V_j(\zeta)$, $j = 1, \dots, 6$, denote the open subsets of $\C$ displayed in Figure \ref{Gamma2.pdf}.

\begin{lemma}[Analytic approximation of $r_1$ and $r_2$]\label{decompositionlemma2}
There exist decompositions
\begin{align*}
& r_1(\lambda) = r_{1,a}(x, t, \lambda) + r_{1,r}(x, t, \lambda), \qquad \lambda > \lambda_0, 
	\\
& r_2(\lambda) = r_{2,a}(x, t, \lambda) + r_{2,r}(x, t, \lambda), \qquad \lambda < \lambda_0, 
\end{align*}
where the functions $\{r_{j,a}, r_{j,r}\}_{j=1}^2$ have the following properties:
\begin{enumerate}[$(a)$]
\item For each $\zeta \in \mathcal{I}$ and each $t > 0$, $r_{j,a}(x, t, \lambda)$ is defined and continuous for $\lambda \in \bar{V}_j$ and analytic for $\lambda \in V_j$, $j = 1,2$.

\item The functions $r_{1,a}$ and $r_{2,a}$ satisfy
\begin{align}\label{rjaestimates}
\begin{cases} 
|r_{j, a}(x, t, \lambda) - r_j(\lambda_0)| \leq C |\lambda - \lambda_0| e^{\frac{t}{4}|\re \Phi(\zeta,\lambda)|}, 
	\\  
|r_{j, a}(x, t, \lambda)| \leq \frac{C}{1 + |\lambda|^2} e^{\frac{t}{4}|\re \Phi(\zeta,\lambda)|},
\end{cases}   \lambda \in \bar{V}_j, \ \zeta \in \mathcal{I}, \ t > 0, \ j = 1, 2, 
\end{align}
where the constant $C$ is independent of $\zeta, t, \lambda$.

\item The $L^1, L^2$, and $L^\infty$ norms on $(\lambda_0, \infty)$ of $r_{1,r}(x, t, \cdot)$ are $O(t^{-3/2})$ as $t \to \infty$ uniformly with respect to $\zeta \in \mathcal{I}$.

\item The $L^1, L^2$, and $L^\infty$ norms on $(-\infty, \lambda_0)$ of $r_{2,r}(x, t, \cdot)$ are $O(t^{-3/2})$ as $t \to \infty$ uniformly with respect to $\zeta \in \mathcal{I}$.
\end{enumerate}
\end{lemma}
\begin{proof}
The proof is similar to that of Lemma \ref{decompositionlemma} and will be omitted, see \cite{DZ1993, Lnonlinearsteepest, XF2012}.
\end{proof}

We introduce $m^{(3)}(x,t,\lambda)$ by
$$m^{(3)}(x,t,\lambda) = m^{(2)}(x,t,\lambda)H(x,t,\lambda),$$
where the sectionally analytic function $H$ is defined by
\begin{align}\label{Hdef}
H = \begin{cases} 
\begin{pmatrix} 1 & 0 \\ \delta^{-2} \lambda r_{1,a} e^{t\Phi}& 1 \end{pmatrix}, & \lambda \in V_1,
	\\
\begin{pmatrix} 1 & -\delta^2 r_{2,a} e^{-t\Phi} \\ 0 & 1 \end{pmatrix}, & \lambda \in V_2,
	\\
\begin{pmatrix} 1 & 0 \\  -\delta^{-2} \lambda r_{2,a}^* e^{t\Phi}& 1 \end{pmatrix}, & \lambda \in V_3,
	\\
\begin{pmatrix} 1 & \delta^2 r_{1,a}^* e^{-t\Phi} \\ 0 & 1 \end{pmatrix}, & \lambda \in V_4,
	\\
\begin{pmatrix} 1 & - \delta^2 h_a^* e^{-t\Phi} \\ 0 & 1 \end{pmatrix}, & \lambda \in V_5,
	\\
\begin{pmatrix} 1 & 0 \\ -\delta^{-2} \lambda h_a e^{t\Phi} & 1 \end{pmatrix}, & \lambda \in V_6,
	\\
I, & \text{elsewhere}.	
\end{cases}
\end{align}
\begin{figure}
\begin{center}
\begin{overpic}[width=.55\textwidth]{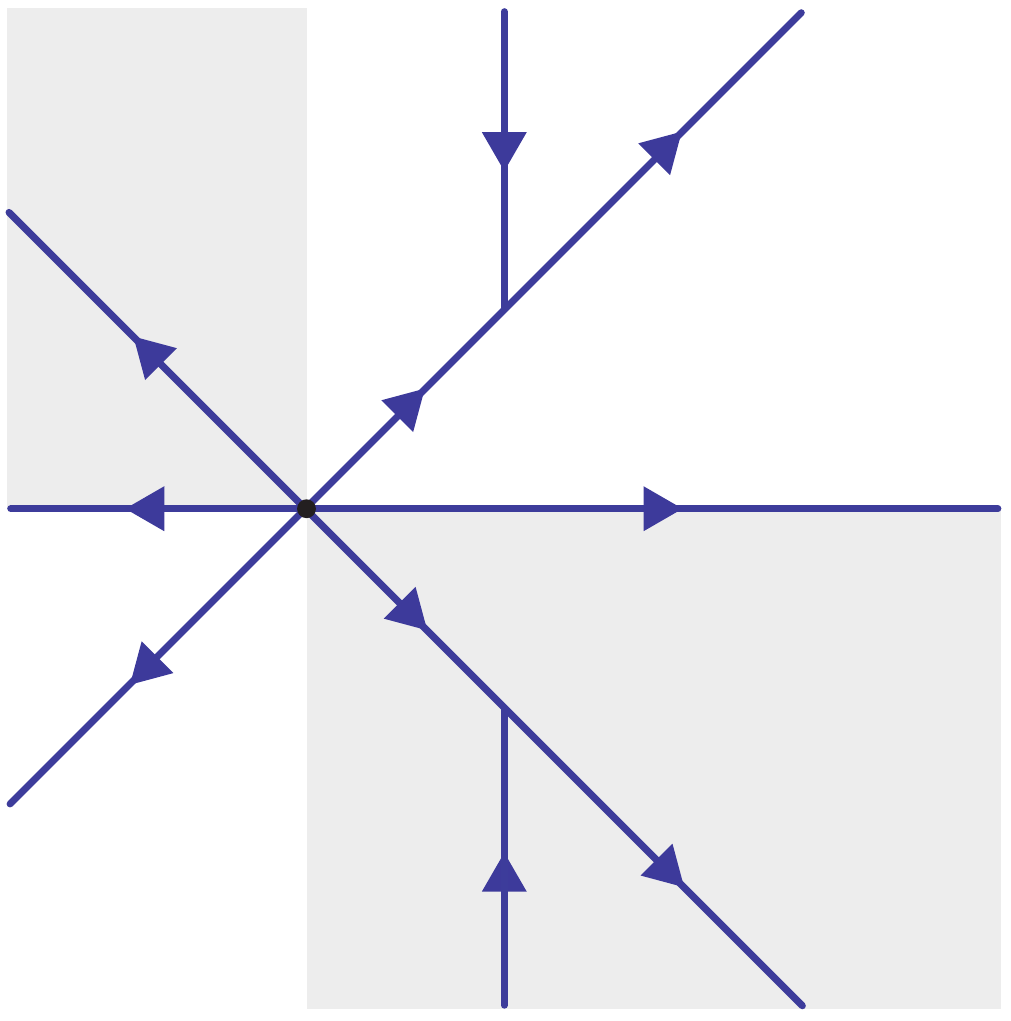}
      \put(101,49){\small $\Gamma^{(3)}$}
      \put(28,44.3){$\lambda_0$}
      \put(75,72){\small $\re \Phi < 0$}
      \put(75,27){\small $\re \Phi > 0$}
      \put(42,57){\small $1$}
      \put(18,64){\small $2$}
      \put(18,33){\small $3$}
      \put(42,40){\small $4$}
      \put(67,15){\small $5$}
      \put(67,82){\small $6$}
      \put(65,52){\small $7$}
      \put(45,83){\small $8$}
      \put(17,52){\small $9$}
      \put(43,14){\small $10$}
      \put(62,65){\small $V_1$}
      \put(7,58){\small $V_2$}
      \put(7,40){\small $V_3$}
      \put(62,34){\small $V_4$}
      \put(56,8){\small $V_5$}
      \put(56,89){\small $V_6$}
      \end{overpic}
     \begin{figuretext}\label{Gamma2.pdf}
       The contour $\Gamma^{(3)}$ and the open sets $\{V_j\}_1^6$ in the complex $\lambda$-plane. The region where $\re \Phi > 0$ is shaded. 
     \end{figuretext}
     \end{center}
\end{figure}
By Lemma \ref{deltalemma}, Lemma \ref{decompositionlemma}, and Lemma \ref{decompositionlemma2}, we have  
$$H(x,t,\cdot)^{\pm1} \in I + (\dot{E}^2 \cap E^\infty)(\C \setminus \Gamma^{(3)}),$$
where $\Gamma^{(3)} \subset \C$ denotes the contour displayed in Figure \ref{Gamma2.pdf}.
It follows that $m$ satisfies the RH problem (\ref{RHm}) if and only if $m^{(3)}$ 
satisfies the RH problem (\ref{RHmj}) with $j = 3$, where the jump matrix $v^{(3)}$ is given by 
\begin{align}\nonumber
&v_1^{(3)} = \begin{pmatrix} 1 & 0 \\ -\delta^{-2} \lambda (r_{1,a} + h) e^{t\Phi}& 1 \end{pmatrix}, 
&&
v_2^{(3)} = \begin{pmatrix} 1 & -\delta^2 r_{2,a} e^{-t\Phi} \\ 0 & 1 \end{pmatrix}, 
	\\\nonumber
&v_3^{(3)} = \begin{pmatrix} 1 & 0 \\  \delta^{-2} \lambda r_{2,a}^* e^{t\Phi}& 1 \end{pmatrix}, 
&&
v_4^{(3)} = \begin{pmatrix} 1 & \delta^2 (r_{1,a}^* + h^*) e^{-t\Phi} \\ 0 & 1 \end{pmatrix}, 
	\\\nonumber
&v_5^{(3)} = \begin{pmatrix} 1 & \delta^2 (r_{1,a}^* + h_a^*) e^{-t\Phi} \\ 0 & 1 \end{pmatrix}, 
&&
v_6^{(3)} = \begin{pmatrix} 1 & 0 \\ -\delta^{-2} \lambda (r_{1,a} + h_a) e^{t\Phi}& 1 \end{pmatrix}, 
	\\\nonumber
&v_7^{(3)} =  \begin{pmatrix} 1-\lambda r_{1,r} r_{1,r}^* & \delta^2 r_{1,r}^* e^{-t\Phi} \\ - \delta^{-2} \lambda r_{1,r} e^{t\Phi}& 1 \end{pmatrix}, 
&&
v_8^{(3)} = \begin{pmatrix} 1 & 0 \\ \delta^{-2} \lambda h_r e^{t\Phi} & 1 \end{pmatrix}, 
	\\ \nonumber
& v_9^{(3)} = \begin{pmatrix} 1 & -\delta_-^2 r_{2,r} e^{-t\Phi} \\ 0 & 1 \end{pmatrix}
\begin{pmatrix} 1 & 0\\  \delta_+^{-2} \lambda r_{2,r}^* e^{t\Phi}& 1 \end{pmatrix}, 
&&
v_{10}^{(3)} = \begin{pmatrix} 1 & - \delta^2 h_r^* e^{-t\Phi} \\ 0 & 1 \end{pmatrix}, 
\end{align}
with $v_j^{(3)}$ denoting the restriction of $v^{(3)}$ to the contour labeled by $j$ in Figure \ref{Gamma2.pdf}.

\section{Local model}\label{localsec}
The transformations of Section \ref{transsec} led to a RH problem for $m^{(3)}$ with the property that the matrix $v^{(3)} - I$ decays to zero as $t \to \infty$ everywhere except near $\lambda_0$. This means that we only have to consider a neighborhood of $\lambda_0$ when computing the long-time asymptotics of $m^{(3)}$. In this section, we find a local solution $m^{\lambda_0}$ which approximates $m^{(3)}$ near $\lambda_0$. The basic idea is that in the large $t$ limit, the RH problem for $m^{(3)}$ near $\lambda_0$ reduces to a RH problem on a cross $X$ which can be solved exactly in terms of parabolic cylinder functions \cite{I1981, DZ1993}.

\subsection{Exact solution on the cross}
Let $X = X_1 \cup \cdots \cup X_4 \subset \C$ be the cross defined by
\begin{align} \nonumber
&X_1 = \bigl\{se^{\frac{i\pi}{4}}\, \big| \, 0 \leq s < \infty\bigr\}, && 
X_2 = \bigl\{se^{\frac{3i\pi}{4}}\, \big| \, 0 \leq s < \infty\bigr\},  
	\\ \label{Xdef}
&X_3 = \bigl\{se^{-\frac{3i\pi}{4}}\, \big| \, 0 \leq s < \infty\bigr\}, && 
X_4 = \bigl\{se^{-\frac{i\pi}{4}}\, \big| \, 0 \leq s < \infty\bigr\},
\end{align}
and oriented away from the origin, see Figure \ref{X.pdf}.

\begin{figure}
\begin{center}
 \begin{overpic}[width=.4\textwidth]{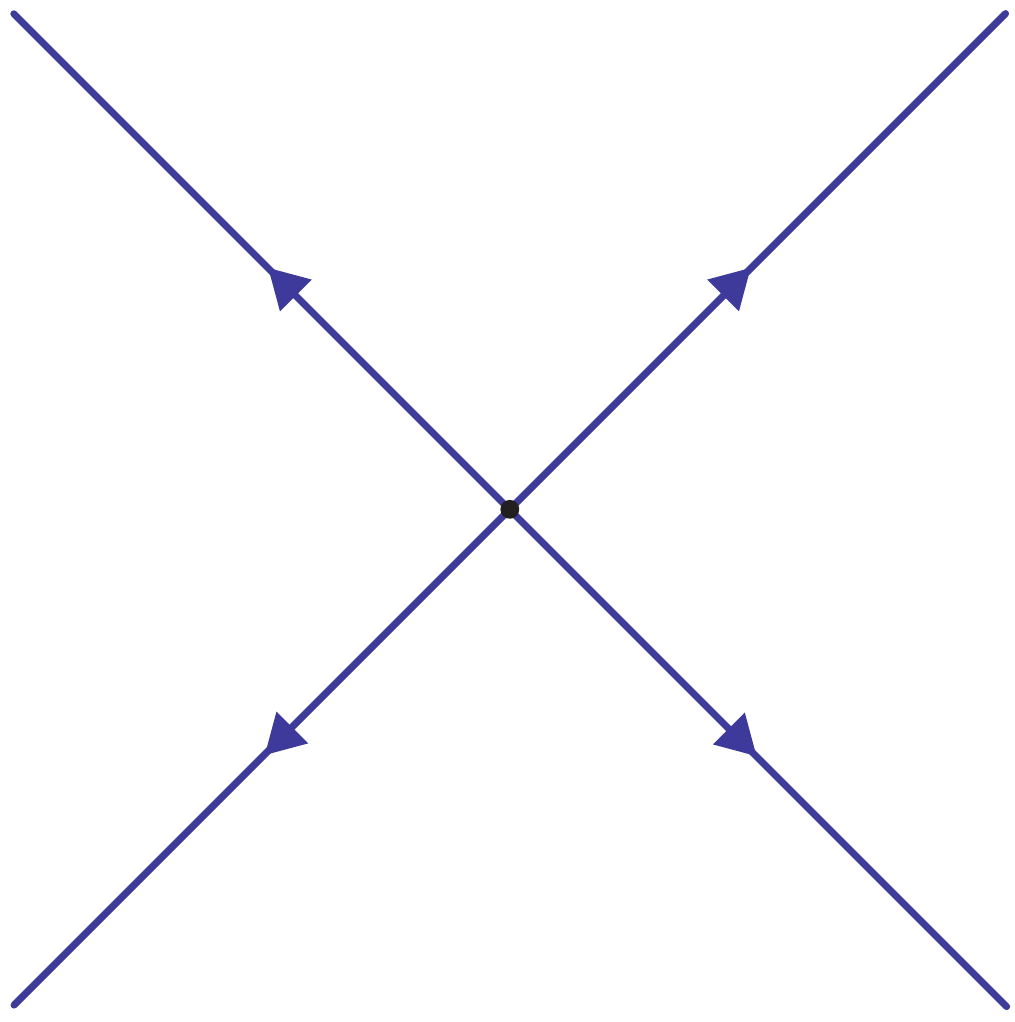}
      \put(73.5,68){\small $X_1$}
      \put(19,68){\small $X_2$}
      \put(17,27){\small $X_3$}
      \put(75,27){\small $X_4$}
      \put(48,43){$0$}
    \end{overpic}
     \begin{figuretext}\label{X.pdf}
        The contour $X = X_1 \cup X_2 \cup X_3 \cup X_4$.
     \end{figuretext}
     \end{center}
\end{figure}

\begin{lemma}[Exact solution on the cross]\label{Xlemma}
Define the function $\nu:\C \to (0,\infty)$ by 
$\nu(q) = \frac{1}{2\pi} \ln(1 + |q|^2)$ and define the jump matrix $v^X(q, z)$ for $z \in X$ by
\begin{align}\label{vXdef} 
v^X(q, z) = \begin{cases}
\begin{pmatrix} 1 & 0	\\
  q z^{2i\nu(q)} e^{\frac{iz^2}{2}}	& 1 \end{pmatrix}, &   z \in X_1, 
  	\\
\begin{pmatrix} 1 & -\frac{\bar{q}}{1 + |q|^2} z^{-2i\nu(q)}e^{-\frac{iz^2}{2}}	\\
0 & 1  \end{pmatrix}, &  z \in X_2, 
	\\
\begin{pmatrix} 1 &0 \\
- \frac{q}{1 + |q|^2}z^{2i\nu(q)} e^{\frac{iz^2}{2}}	& 1 \end{pmatrix}, &  z \in X_3,
	\\
 \begin{pmatrix} 1	& \bar{q} z^{-2i\nu(q)}e^{-\frac{iz^2}{2}}	\\
0	& 1 \end{pmatrix}, &  z \in X_4.
\end{cases}
\end{align}
Then, for each $q \in \C$, the RH problem 
\begin{align*}
\begin{cases} m^X(q, \cdot) \in I + \dot{E}^2(\C \setminus X), 
	\\
m_+^X(q, z) =  m_-^X(q, z) v^X(q, z) \quad \text{for a.e.} \ z \in X, 
\end{cases} 
\end{align*}
has a unique solution $m^X(q, z)$. This solution satisfies
\begin{align}\label{mXasymptotics}
  m^X(q, z) = I - \frac{i}{z}\begin{pmatrix} 0 & \beta^X(q) \\ \overline{\beta^X(q)} & 0 \end{pmatrix} + O\biggl(\frac{q}{z^2}\biggr), \qquad z \to \infty,  \ q \in \C, 
\end{align}  
where the error term is uniform with respect to $\arg z \in [0, 2\pi]$ and $q$ in compact subsets of $\C$, and the function $\beta^X(q)$ is defined by
\begin{align}\label{betacdef}
\beta^X(q) = \sqrt{\nu(q)} e^{i\left(\frac{\pi}{4} - \arg q - \arg \Gamma(i\nu(q)\right)}, \qquad q \in \C.
\end{align}
Moreover, for each compact subset $K$ of $\C$, 
$$\sup_{q \in K} \sup_{z \in \C \setminus X} |m^X(q, z)| < \infty$$
and
\begin{align}\label{mXqbound}
\sup_{q \in K} \sup_{z \in \C \setminus X} \frac{|m^X(q, z)- I|}{|q|} < \infty.
\end{align}
\end{lemma}
\begin{proof}
The proof relies on deriving an explicit formula for the solution $m^X$ in terms of parabolic cylinder functions \cite{I1981}. The lemma is standard except possibly for the presence of $q$ in the error term in (\ref{mXasymptotics}) and for the estimate (\ref{mXqbound}).  
The estimate (\ref{mXqbound}) can be derived by analyzing the behavior of the explicit formula for $m^X$ as $q \to 0$ (or by noting that the $L^1$ and $L^\infty$ norms of $v - I$ are $O(q)$ as $q \to 0$). 
Using that $\beta^X(q) = O(q)$ as $q \to 0$, the error term in (\ref{mXasymptotics}) also follows from the explicit formula.
\end{proof}

\begin{remark}\upshape
The estimate (\ref{mXqbound}), which shows that $m^X$ approaches the identity matrix as $q$ goes to zero, will be important for the subsequent analysis. Indeed, for the analysis of the derivative NLS equation, we will apply Lemma \ref{Xlemma} with $q = q(\zeta)$ such that $q \to 0$ as $\zeta = x/t \to 0$, see equation (\ref{qdef}). Therefore, in order to find the asymptotics of $u(x,t)$ near the boundary at $x = 0$, it is crucial to have good error estimates as $q \to 0$.
The inclusion of $q$ in the error term in (\ref{mXasymptotics}) is important for the same reason.
\end{remark}

\subsection{Local model near $\lambda_0$}
Fix a small $\epsilon > 0$ and let $D_\epsilon(\lambda_0)$ denote the open disk of radius $\epsilon$ centered at $\lambda_0$.
In order to relate $m^{(3)}$ to the solution $m^X$ of Lemma \ref{Xlemma}, we make a local change of variables for $\lambda$ near $\lambda_0$ and introduce the new variable $z := z(\zeta, \lambda)$ by 
$$z = \sqrt{8t}(\lambda -\lambda_0).$$
For each $\zeta \in \mathcal{I}$, the map $\lambda \mapsto z$ is a biholomorphism from $D_\epsilon(\lambda_0)$ to the open disk of radius $\sqrt{8t} \epsilon$ centered at the origin. 
We note that $z$ satisfies 
$$\frac{iz^2}{2} = 4it(\lambda - \lambda_0)^2 = t(\Phi(\zeta, \lambda) - \Phi(\zeta, \lambda_0)).$$ 

Integration by parts gives
\begin{align*}
 \int_{-\infty}^{\lambda_0}  & \frac{\ln(1 - s |r(s)|^2)}{s-\lambda} ds
= \ln (\lambda-s)\ln(1 - s |r(s)|^2)\big|_{s=-\infty}^{\lambda_0} 
	\\
& - \int_{-\infty}^{\lambda_0} \ln (\lambda-s) d\ln(1 - s |r(s)|^2), \qquad \zeta \in \mathcal{I}, \ \lambda \in \C \setminus (-\infty, \lambda_0].
\end{align*}

It follows that
\begin{align}\label{deltanuchi}
\delta(\zeta, \lambda) = e^{-i \nu \ln (\lambda - \lambda_0) + \chi(\zeta, \lambda)},
\end{align}
where  $\nu := \nu(\zeta) \geq 0$ is defined by
\begin{align}\label{nudef}
\nu(\zeta) = \frac{1}{2\pi} \ln(1 - \lambda_0 |r(\lambda_0)|^2), \qquad \zeta \in \mathcal{I},
\end{align}
and the function $\chi(\zeta, \lambda)$ is given by
\begin{align}\label{chidef}
\chi(\zeta, \lambda) = -\frac{1}{2\pi i} \int_{-\infty}^{\lambda_0} \ln (\lambda-s) d\ln(1 - s|r(s)|^2), \qquad \zeta \in \mathcal{I}, \ \lambda \in \C \setminus (-\infty, \lambda_0).
\end{align}
Hence we can write $\delta$ as
$$\delta(\zeta, \lambda)  = z^{-i\nu} \delta_0(\zeta, t) \delta_1(\zeta, \lambda), \qquad \zeta \in \mathcal{I}, \ \lambda \in \C \setminus (-\infty, \lambda_0),$$
where the functions $\delta_0(\zeta, t)$ and $\delta_1(\zeta, \lambda)$ are defined by
\begin{align*}
& \delta_0(\zeta, t) = (8t)^{\frac{i\nu}{2}} e^{\chi(\zeta, \lambda_0)}, \qquad t > 0,
	\\ 
&\delta_1(\zeta, \lambda) = e^{\chi(\zeta, \lambda) - \chi(\zeta, \lambda_0)}, \qquad \lambda \in D_\epsilon(\lambda_0).
\end{align*}

Define $\tilde{m}(x,t,z)$ by
$$\tilde{m}(x,t,z(\zeta, \lambda)) = m^{(3)}(x,t,\lambda)e^{-\frac{t\Phi(\zeta, \lambda_0)\sigma_3}{2}}\delta_0(\zeta,t)^{\sigma_3}|\lambda_0|^{-\frac{\sigma_3}{4}}, \qquad \lambda \in \C \setminus \Gamma^{(3)}.$$
Let $\mathcal{X}:= \mathcal{X}(\zeta) = \lambda_0 + X$ denote the cross $X$ centered at $\lambda_0$, see Figure \ref{Xcal.pdf}.
Then $\tilde{m}$ is a sectionally analytic function of $z$ which satisfies $\tilde{m}_+ = \tilde{m}_- \tilde{v}$ on $X$, where the jump matrix 
$$\tilde{v} = |\lambda_0|^{\frac{\hat{\sigma}_3}{4}} \delta_0(\zeta, t)^{-\hat{\sigma}_3}e^{\frac{t\Phi(\zeta, \lambda_0)}{2}\hat{\sigma}_3} v^{(3)}$$ 
is given for $z \in X$ by
\begin{align}\nonumber
&\tilde{v}(x,t,z) = \begin{cases}
\begin{pmatrix} 1 & 0 \\ - z^{2i\nu} \delta_1^{-2} \frac{\lambda}{|\lambda_0|^{1/2}} (r_{1,a} + h_a) e^{\frac{iz^2}{2}}& 1 \end{pmatrix}, & \lambda \in \mathcal{X}_1 \cap \mathcal{D}_1,
	\\
\begin{pmatrix} 1 & 0 \\ - z^{2i\nu} \delta_1^{-2} \frac{\lambda}{|\lambda_0|^{1/2}} (r_{1,a} + h) e^{\frac{iz^2}{2}}& 1 \end{pmatrix}, & \lambda \in \mathcal{X}_1 \cap \mathcal{D}_2,
	\\
\begin{pmatrix} 1 & -z^{-2i\nu} \delta_1^2 |\lambda_0|^{1/2}r_{2,a} e^{-\frac{iz^2}{2}} \\ 0 & 1 \end{pmatrix}, & \lambda \in \mathcal{X}_2,
	\\
\begin{pmatrix} 1 & 0 \\ z^{2i\nu} \delta_1^{-2} \frac{\lambda r_{2,a}^*}{|\lambda_0|^{1/2}} e^{\frac{iz^2}{2}}& 1 \end{pmatrix}, & \lambda \in \mathcal{X}_3,
	\\
\begin{pmatrix} 1 & z^{-2i\nu} \delta_1^2 |\lambda_0|^{1/2}(r_{1,a}^* + h_a^*) e^{-\frac{iz^2}{2}} \\ 0 & 1 \end{pmatrix}, & \lambda \in \mathcal{X}_4 \cap \mathcal{D}_4,
	\\
\begin{pmatrix} 1 & z^{-2i\nu} \delta_1^2 |\lambda_0|^{1/2}(r_{1,a}^* + h^*) e^{-\frac{iz^2}{2}} \\ 0 & 1 \end{pmatrix}, & \lambda \in \mathcal{X}_4 \cap \mathcal{D}_3.
\end{cases}
\end{align}

Define 
\begin{align}\label{qdef}
  q := q(\zeta) = |\lambda_0|^{1/2} r(\lambda_0), \qquad \zeta \in \mathcal{I}.
\end{align}
For a fixed $z$, $|\lambda|^{1/2}r(\lambda) \to q$ and $\delta_1(\lambda(\zeta, z)) \to 1$ as $t\to \infty$. This suggests that $\tilde{v}$ tends to the jump matrix $v^{X}$ defined in (\ref{vXdef}) for large $t$. 
In other words, that the jumps of $m^{(3)}$ for $\lambda$ near $\lambda_0$ approach those of the function $m^X \delta_0^{-\sigma_3} |\lambda_0|^{\frac{\sigma_3}{4}} e^{\frac{t\Phi(\zeta,\lambda_0)\sigma_3}{2}}$  as $t \to \infty$. 
This suggests that we approximate $m^{(3)}$ in the neighborhood $D_\epsilon(\lambda_0)$ of $\lambda_0$ by a $2 \times 2$-matrix valued function $m^{\lambda_0}$ of the form
\begin{align}\label{mmudef}
m^{\lambda_0}(x,t,\lambda) = Y(\zeta,t,\lambda) m^X(q(\zeta),z(\zeta, \lambda)) \delta_0(\zeta, t)^{-\sigma_3} |\lambda_0|^{\frac{\sigma_3}{4}} e^{\frac{t\Phi(\zeta, \lambda_0)\sigma_3}{2}},
\end{align}
where $Y(\zeta,t,\lambda)$ is a function which is analytic for $\lambda \in D_\epsilon(\lambda_0)$.
To ensure that $m^{\lambda_0}$ is a good approximation of $m^{(3)}$ for large $t$, we choose $Y$ so that $m^{\lambda_0} \to I$ on $\partial D_\epsilon(\lambda_0)$ as $t \to \infty$. 
Hence we choose
\begin{align}\label{Ymudef}
Y(\zeta,t,\lambda) \equiv
Y(\zeta,t) = e^{-\frac{t\Phi(\zeta, \lambda_0)\sigma_3}{2}}|\lambda_0|^{-\frac{\sigma_3}{4}} \delta_0(\zeta, t)^{\sigma_3}.
\end{align}

\begin{figure}
\begin{center}
  \begin{overpic}[width=.55\textwidth]{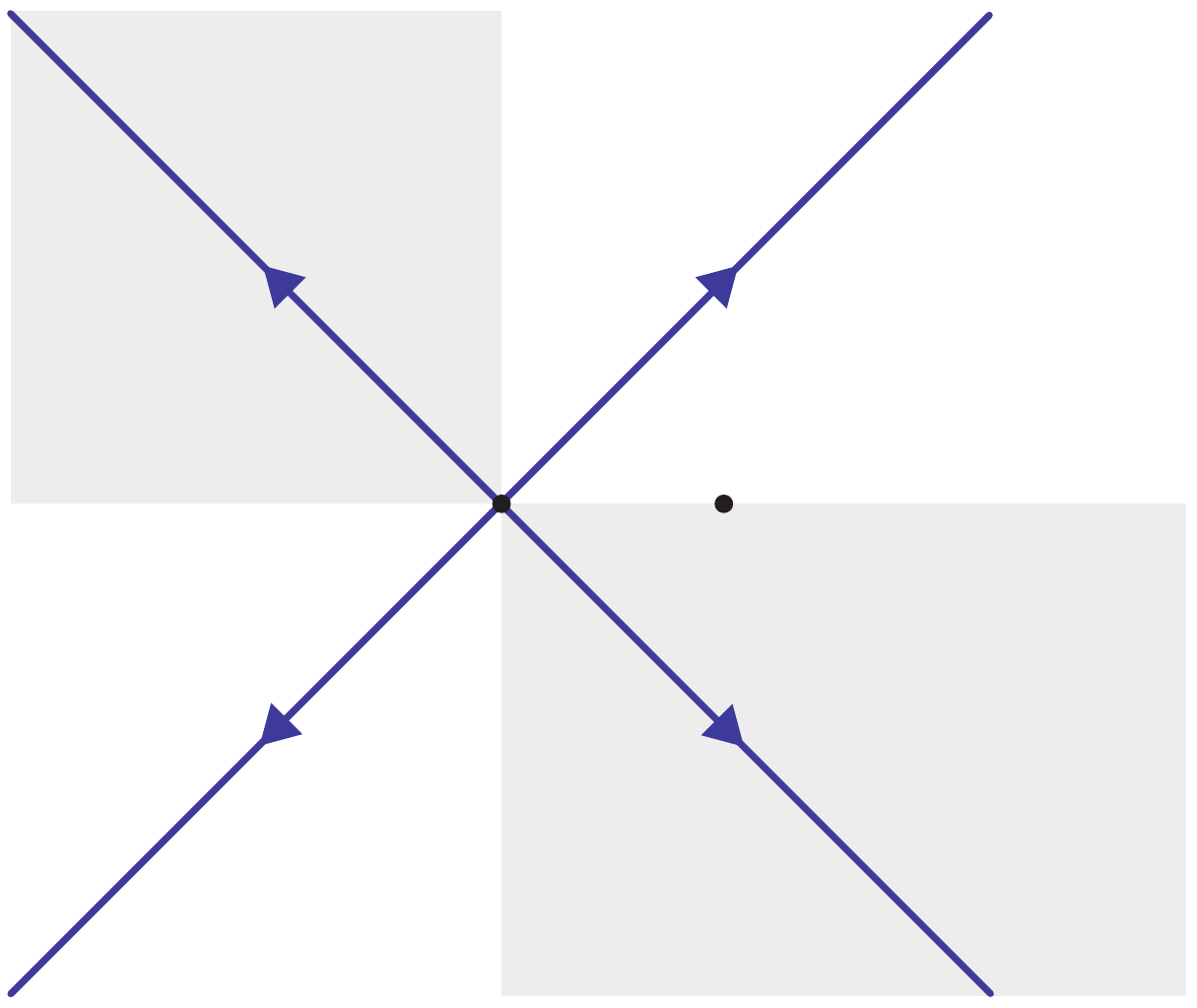}
      \put(62,58){\small $\mathcal{X}_1$}
      \put(16.5,58){\small $\mathcal{X}_2$}
      \put(16,23){\small $\mathcal{X}_3$}
      \put(63,22){\small $\mathcal{X}_4$}
      \put(40,37){\small $\lambda_0$}
      \put(59.5,37){\small $0$}
    \end{overpic}
     \begin{figuretext}\label{Xcal.pdf}
        The contour $\mathcal{X} = \mathcal{X}_1 \cup \mathcal{X}_2 \cup \mathcal{X}_3 \cup \mathcal{X}_4$.
     \end{figuretext}
     \end{center}
\end{figure}
Let $\mathcal{X}^\epsilon := \mathcal{X}^\epsilon(\zeta)$ denote the part of $\mathcal{X}$ that lies in the disk $D_\epsilon(\lambda_0)$, i.e., $\mathcal{X}^\epsilon = \mathcal{X} \cap D_\epsilon(\lambda_0)$.
In the following $C$ denotes a generic constant independent of $\zeta, t, \lambda$, which may change within a computation.

\begin{lemma}\label{mulemma}
For each  $\zeta \in \mathcal{I}$ and $t > 0$, the function $m^{\lambda_0}(x,t,\lambda)$ defined in (\ref{mmudef}) is an analytic function of $\lambda \in D_\epsilon(\lambda_0) \setminus \mathcal{X}^\epsilon$. Moreover,
\begin{align}\label{mlambda0bound}
|m^{\lambda_0}(x,t,\lambda) - I| \leq C\frac{|q|}{|\lambda_0|^{1/2}} \leq C, \qquad \zeta \in\mathcal{I}, \ t > 2, \ \lambda \in \overline{D_\epsilon(\lambda_0)} \setminus \mathcal{X}^\epsilon.
\end{align}
Across $\mathcal{X}^\epsilon$, $m^{\lambda_0}$ obeys the jump condition $m_+^{\lambda_0} =  m_-^{\lambda_0} v^{\lambda_0}$, where the jump matrix $v^{\lambda_0}$ satisfies
\begin{align}\label{v5vmuestimate}
\begin{cases}
 \|v^{(3)} - v^{\lambda_0}\|_{L^1(\mathcal{X}^\epsilon)} \leq C t^{-1} \ln t, 
	\\
 \|v^{(3)} - v^{\lambda_0}\|_{L^2(\mathcal{X}^\epsilon)} \leq C t^{-3/4} \ln t, 
	\\
\|v^{(3)} - v^{\lambda_0}\|_{L^\infty(\mathcal{X}^\epsilon)} \leq  C t^{-1/2} \ln t,
\end{cases} \qquad \zeta \in \mathcal{I}, \ t >2.
\end{align}	
Furthermore, as $t \to \infty$,
\begin{align}\label{mmodmuestimate2}
\|m^{\lambda_0}(x,t,\cdot)^{-1} - I\|_{L^\infty(\partial D_\epsilon(\lambda_0))} = O\bigg(\frac{q}{|\lambda_0|^{1/2} \sqrt{t}}\bigg), 
\end{align}
and
\begin{align}\label{mmodmuestimate1}
\frac{1}{2\pi i}\int_{\partial D_\epsilon(\lambda_0)}(m^{\lambda_0}(x,t,\lambda)^{-1} - I) d\lambda
= -\frac{Y(\zeta,t) m_1^X(\zeta) Y(\zeta,t)^{-1}}{\sqrt{8t}} + O\bigg(\frac{q}{|\lambda_0|^{1/2} t}\bigg),
\end{align}	
uniformly with respect to $\zeta \in \mathcal{I}$, where $m_1^X(\zeta)$ is defined by
\begin{align}\label{m1Xdef}
m_1^X(\zeta) = -i\begin{pmatrix} 0 & \beta^X(q(\zeta)) \\ \overline{\beta^X(q(\zeta))} & 0 \end{pmatrix}.
\end{align}
\end{lemma}
\begin{proof}
The analyticity of $m^{\lambda_0}$ follows directly from the definition. Since
$$m^{\lambda_0}(x,t,\lambda) = Y(\zeta,t)m^X(q, z)Y(\zeta,t)^{-1} 
= e^{-\frac{t\Phi(\zeta, \lambda_0)\hat{\sigma}_3}{2}}|\lambda_0|^{-\frac{\hat{\sigma}_3}{4}} \delta_0(\zeta, t)^{\hat{\sigma}_3} m^X(q, z),$$
where $|e^{-\frac{t\Phi(\zeta, \lambda_0)}{2}}| = |\delta_0(\zeta, t)| = 1$, the estimate (\ref{mlambda0bound}) is a direct consequence of (\ref{mXqbound}).

We next establish (\ref{v5vmuestimate}). Standard estimates show that
$$|\chi(\zeta, \lambda) - \chi(\zeta, \lambda_0)| \leq  C |\lambda - \lambda_0| ( 1+ |\ln|\lambda-\lambda_0||), \qquad  \zeta \in \mathcal{I}, \ \lambda \in \mathcal{X}^\epsilon.$$
This yields
\begin{align}\nonumber
|\delta_1(\zeta, \lambda) - 1| & =  |e^{\chi(\zeta, \lambda) - \chi(\zeta, \lambda_0)} -1|
	\\ \label{delta1estimate}
& \leq  C |\lambda - \lambda_0| ( 1+ |\ln|\lambda-\lambda_0||), \qquad \zeta \in \mathcal{I}, \ \lambda \in \mathcal{X}^\epsilon.
\end{align}
On the other hand, equations (\ref{haestimate}) and (\ref{rjaestimates}) imply that the following estimates hold for all $\zeta \in \mathcal{I}$ and $t > 2$:
\begin{align}\nonumber
& \big|\lambda(r_{1,a}(x,t,\lambda) + h_a(t,\lambda)) + |\lambda_0|^{1/2}q\big| \leq C|\lambda - \lambda_0|e^{\frac{t}{4} |\re \Phi(\zeta,\lambda )|}, && \lambda \in \mathcal{X}_1^\epsilon \cap \mathcal{D}_1,
 	\\\nonumber
& \big|\lambda(r_{1,a}(x,t,\lambda) + h(\lambda)) + |\lambda_0|^{1/2}q\big| \leq C|\lambda - \lambda_0|e^{\frac{t}{4} |\re \Phi(\zeta,\lambda )|}, && \lambda \in \mathcal{X}_1^\epsilon \cap \mathcal{D}_2,
 	\\\nonumber
& \bigg|r_{2,a}(x,t,\lambda) - |\lambda_0|^{-1/2}\frac{\bar{q}}{1 + |q|^2}\bigg| \leq C|\lambda - \lambda_0|e^{\frac{t}{4} |\re \Phi(\zeta,\lambda )|},&& \lambda \in \mathcal{X}_2^\epsilon,
  	\\\nonumber
& \bigg| \lambda \overline{r_{2,a}(x,t,\bar{\lambda})} + |\lambda_0|^{1/2}\frac{q}{1 + |q|^2}\bigg| \leq C|\lambda - \lambda_0|e^{\frac{t}{4} |\re \Phi(\zeta,\lambda )|},&& \lambda \in \mathcal{X}_3^\epsilon,
 	\\ \nonumber
&  \big|\overline{r_{1,a}(x,t,\bar{\lambda})} + \overline{h_a(t,\bar{\lambda})} - |\lambda_0|^{-1/2}\bar{q}\big| \leq C|\lambda - \lambda_0|e^{\frac{t}{4} |\re \Phi(\zeta,\lambda )|}, 
&& \lambda \in \mathcal{X}_4^\epsilon \cap \mathcal{D}_4,
 	\\ \label{rqestimate}
&  \big| \overline{r_{1,a}(x,t,\bar{\lambda})} + \overline{h(\bar{\lambda})} - |\lambda_0|^{-1/2}\bar{q}\big| \leq C|\lambda - \lambda_0|e^{\frac{t}{4} |\re \Phi(\zeta,\lambda )|}, 
&& \lambda \in \mathcal{X}_4^\epsilon \cap \mathcal{D}_3.
\end{align}
Indeed, for $\lambda \in \mathcal{X}_1^\epsilon \cap \mathcal{D}_1$, the estimates (\ref{haestimate}) and (\ref{rjaestimates}) give
\begin{align*}
 \big|\lambda(&r_{1,a}(x,t,\lambda) + h_a(t,\lambda)) + |\lambda_0|^{1/2}q\big| 
= \big|\lambda(r_{1,a}(x,t,\lambda) + h_a(t,\lambda)) - \lambda_0 (r_1(\lambda_0) + h(\lambda_0))\big| 
	\\
 \leq & \; |\lambda - \lambda_0||r_{1,a}(x,t,\lambda) + h_a(t,\lambda)| + |\lambda_0| |r_{1,a}(x,t,\lambda)- r_1(\lambda_0)| + |\lambda_0||h_a(t,\lambda) - h(0)| 
	\\
& + |\lambda_0||h(0) - h(\lambda_0)|
	\\
\leq &\; C |\lambda - \lambda_0|\frac{e^{\frac{t}{4} |\re \Phi(\zeta,\lambda )|}}{1 + |\lambda|^2}
+ C|\lambda_0| |\lambda - \lambda_0|e^{\frac{t}{4} |\re \Phi(\zeta,\lambda )|}
 + C|\lambda_0||\lambda|e^{\frac{t}{4} |\re \Phi(\zeta,\lambda )|} + C|\lambda_0|^2
	\\	 
\leq &\; C|\lambda - \lambda_0|e^{\frac{t}{4} |\re \Phi(\zeta,\lambda )|},
\end{align*}
which proves the first estimate in (\ref{rqestimate}); the proofs of the other estimates are similar.

Since
\begin{align*}|\lambda_0|^{-\frac{\hat{\sigma}_3}{4}}(\tilde{v} - v^X)
= 
 \begin{cases}
\begin{pmatrix} 0 & 0 \\ -(\delta_1^{-2} \lambda(r_{1,a} + h_a) + |\lambda_0|^{1/2} q )z^{2i\nu}e^{\frac{iz^2}{2}} & 0 \end{pmatrix}, & \lambda \in \mathcal{X}_1 \cap \mathcal{D}_1,
	\\
\begin{pmatrix} 0 & 0 \\ -(\delta_1^{-2} \lambda(r_{1,a} + h) + |\lambda_0|^{1/2}q )z^{2i\nu}e^{\frac{iz^2}{2}} & 0 \end{pmatrix}, & \lambda \in \mathcal{X}_1 \cap \mathcal{D}_2,
	\\
\begin{pmatrix} 0 & (-\delta_1^2 r_{2,a} + |\lambda_0|^{-1/2}\frac{\bar{q}}{1 + |q|^2})z^{-2i\nu} e^{-\frac{iz^2}{2}} \\ 0 & 0 \end{pmatrix}, & \lambda \in \mathcal{X}_2,
	\\
\begin{pmatrix} 0 & 0 \\ (\delta_1^{-2} \lambda r_{2,a}^* + |\lambda_0|^{1/2}\frac{q}{1 + |q|^2})z^{2i\nu} e^{\frac{iz^2}{2}} & 0 \end{pmatrix}, & \lambda \in \mathcal{X}_3,
	\\
\begin{pmatrix} 0 & (\delta_1^2 (r_{1,a}^* + h_a^*) - |\lambda_0|^{-1/2}\bar{q}) z^{-2i\nu}e^{-\frac{iz^2}{2}}\\ 0 & 0 \end{pmatrix}, & \lambda \in \mathcal{X}_4 \cap \mathcal{D}_4,
	\\
\begin{pmatrix} 0 & (\delta_1^2 (r_{1,a}^* + h^*) - |\lambda_0|^{-1/2}\bar{q}) z^{-2i\nu}e^{-\frac{iz^2}{2}}\\ 0 & 0 \end{pmatrix}, & \lambda \in \mathcal{X}_4 \cap \mathcal{D}_3.
\end{cases}
\end{align*}
equations (\ref{delta1estimate}) and (\ref{rqestimate}) imply 
\begin{align}\label{vtildevX}
||\lambda_0|^{-\frac{\hat{\sigma}_3}{4}}(\tilde{v} - v^X)| \leq C|\lambda - \lambda_0| ( 1+ |\ln|\lambda-\lambda_0||)e^{-3t|\lambda - \lambda_0|^2}, \qquad \lambda \in \mathcal{X}^\epsilon.
\end{align}
Indeed, for $\lambda \in  \mathcal{X}_1 \cap \mathcal{D}_1$, we have
$$\re\Big(\frac{iz^2}{2}\Big) = t \re \Phi(\zeta, \lambda) = -4t|\lambda - \lambda_0|^2,$$ 
and hence
\begin{align*}
 ||\lambda_0|^{-\frac{\hat{\sigma}_3}{4}}(\tilde{v} - v^X)|
\leq &\; |\delta_1^{-2} - 1| \big|\lambda(r_{1,a} + h_a)z^{2i\nu}e^{\frac{iz^2}{2}}\big|
+ \big|(\lambda(r_{1,a} + h_a) + |\lambda_0|^{1/2} q )z^{2i\nu}e^{\frac{iz^2}{2}}\big|
	\\
\leq&\; C |\delta_1^{-2} - 1| |\lambda| e^{\frac{t}{4} |\re \Phi(\zeta,\lambda )|} e^{\re \frac{iz^2}{2}}
+ C|\lambda - \lambda_0|e^{\frac{t}{4} |\re \Phi(\zeta,\lambda )|} e^{\re \frac{iz^2}{2}}
	\\
\leq &\;
C |\lambda - \lambda_0| ( 1+ |\ln|\lambda-\lambda_0||)|\lambda|e^{-3t|\lambda - \lambda_0|^2}
	\\
& + C |\lambda - \lambda_0| e^{-3t|\lambda - \lambda_0|^2}, \qquad \lambda \in  \mathcal{X}_1 \cap \mathcal{D}_1,
\end{align*}
which gives (\ref{vtildevX}) for $\lambda \in \mathcal{X}_1^\epsilon \cap \mathcal{D}_1$; the proof is similar for the other parts of $\mathcal{X}^\epsilon$.

Since
$$v^{(3)} - v^{\lambda_0} =  e^{-\frac{t\Phi(\zeta, \lambda_0)}{2}\hat{\sigma}_3} \delta_0^{\hat{\sigma}_3}|\lambda_0|^{-\frac{\hat{\sigma}_3}{4}}(\tilde{v} - v^X), \qquad \lambda \in \mathcal{X}^\epsilon,$$
we arrive at
$$|v^{(3)} - v^{\lambda_0}| \leq C|\lambda - \lambda_0| ( 1+ |\ln|\lambda-\lambda_0||)e^{-3t|\lambda-\lambda_0|^2}, \qquad \lambda \in \mathcal{X}^\epsilon.$$
Thus
$$ \|v^{(3)} - v^{\lambda_0}\|_{L^1(\mathcal{X}^\epsilon)}
\leq C \int_0^\epsilon u ( 1+ |\ln u|)e^{-3tu^2} du
\leq C t^{-1} \ln t, \qquad \zeta \in \mathcal{I}, \ t > 2,$$
and
$$ \|v^{(3)} - v^{\lambda_0}\|_{L^\infty(\mathcal{X})}
\leq C \sup_{0 \leq u \leq \epsilon} u ( 1+ |\ln u|)e^{-3 tu^2} 
\leq  C t^{-1/2} \ln t, \qquad \zeta \in \mathcal{I}, \  t > 2,$$
which gives (\ref{v5vmuestimate}).

The variable $z = \sqrt{8t}(\lambda - \lambda_0)$ goes to infinity as $t\to \infty$ if $\lambda \in \partial D_\epsilon(\lambda_0)$.
Thus equation (\ref{mXasymptotics}) yields
\begin{align*}
 & m^X(q, z(\zeta, \lambda)) = I + \frac{m_1^X(\zeta)}{\sqrt{8t}(\lambda - \lambda_0)} + O(qt^{-1}), 
\qquad  t \to \infty,  \ \lambda \in \partial D_\epsilon(\lambda_0),
 \end{align*}  
uniformly with respect to $\lambda \in \partial D_\epsilon(\lambda_0)$, where $m_1^X$ is given by (\ref{m1Xdef}).
Since
$$m^{\lambda_0}(x,t,\lambda) = Y m^X(q, z(\zeta, \lambda)) Y^{-1}$$
and $|Y| \leq C |\lambda_0|^{-1/4}$, this shows that
\begin{align}\label{mmodmmuI}
 & (m^{\lambda_0})^{-1} - I = - \frac{Ym_1^X(\zeta) Y^{-1}}{\sqrt{8t}(\lambda - \lambda_0)} + O\bigg(\frac{q}{|\lambda_0|^{1/2}t}\bigg), \qquad  t \to \infty,
 \end{align}  
uniformly with respect to $\zeta \in \mathcal{I}$ and $\lambda \in \partial D_\epsilon(\lambda_0)$. Using that $|m_1^X| \leq C |q|$, this proves (\ref{mmodmuestimate2}). 
Finally, equation (\ref{mmodmuestimate1})  follows from (\ref{mmodmmuI}) and Cauchy's formula.
\end{proof}

\section{Final steps}\label{finalsec}

Define the approximate solution $m^{app}$ by
$$m^{app} = \begin{cases} m^{\lambda_0}, & \lambda \in D_\epsilon(\lambda_0), \\
I, & \text{elsewhere}.
\end{cases}$$
We will show that the solution $\hat{m}(x,t,\lambda)$ defined by
$$\hat{m} = m^{(3)} (m^{app} )^{-1}.$$
is small for large $t$. The function $\hat{m}$ satisfies the RH problem
\begin{align}\label{RHmhat}
\begin{cases}
\hat{m}(x, t, \cdot) \in I + \dot{E}^2(\C \setminus \hat{\Gamma}),\\
\hat{m}_+(x,t,\lambda) = \hat{m}_-(x, t, \lambda) \hat{v}(x, t, \lambda) \quad \text{for a.e.} \ \lambda \in \hat{\Gamma},
\end{cases}
\end{align}
where the contour $\hat{\Gamma} = \Gamma^{(3)} \cup \partial D_\epsilon(\lambda_0)$ is displayed in Figure \ref{Gammahat.pdf} and the jump matrix $\hat{v}$ is given by
$$\hat{v} = \begin{cases}
v^{(3)}, & \lambda \in \hat{\Gamma} \setminus \overline{D_\epsilon(\lambda_0)}, 
	\\
(m^{\lambda_0})^{-1}, & \lambda \in \partial D_\epsilon(\lambda_0), 
	\\
m_-^{\lambda_0} v^{(3)}(m_+^{\lambda_0})^{-1}, & \lambda \in \hat{\Gamma} \cap D_\epsilon(\lambda_0).
\end{cases}
$$
\begin{figure}
\begin{center}
 \begin{overpic}[width=.5\textwidth]{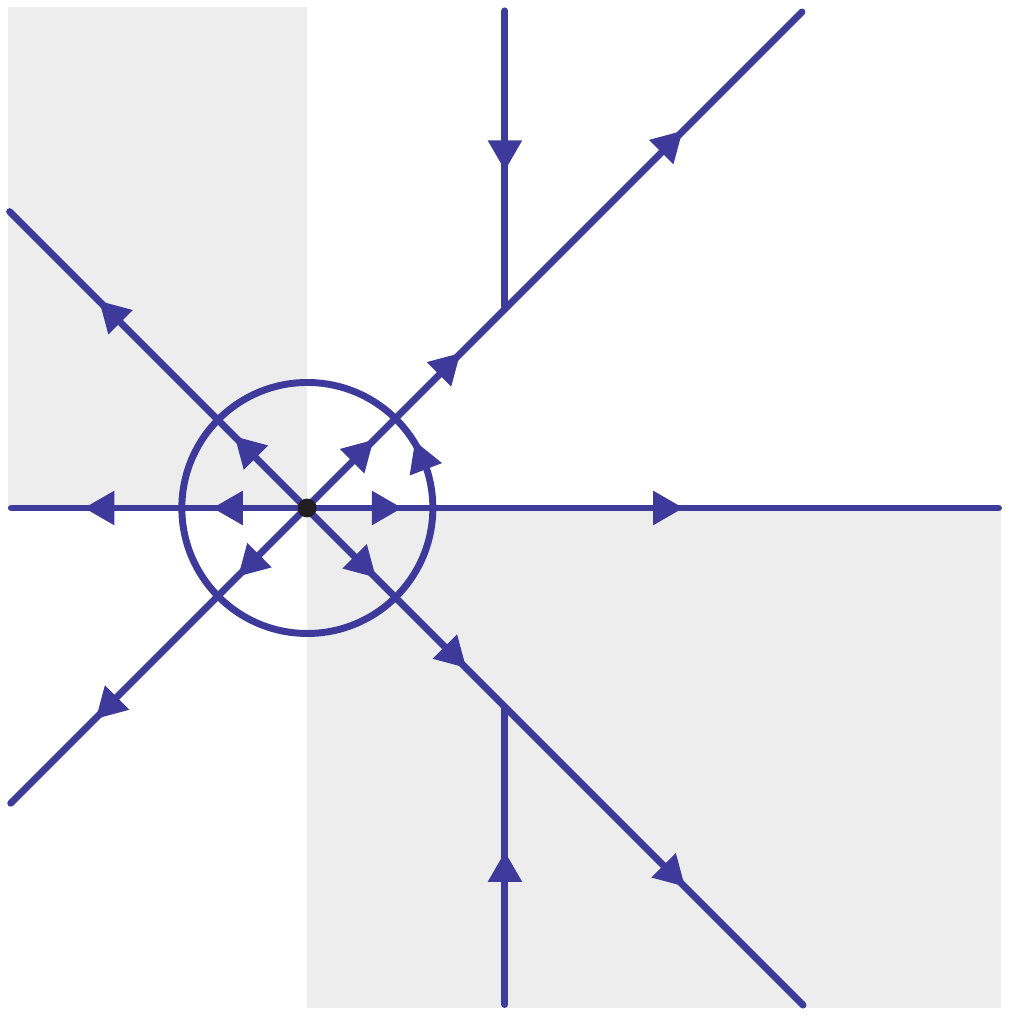}
      \put(102,48){$\hat{\Gamma}$}
      \put(28,44.3){$\lambda_0$}
      \put(68,68){\small $\re \Phi < 0$}
      \put(68,30){\small $\re \Phi > 0$}
    \end{overpic}
    \vspace{.5cm}
     \begin{figuretext}\label{Gammahat.pdf}
        The contour  $\hat{\Gamma}$ together with the domains where $\re \Phi < 0$ (white) and $\re \Phi > 0$ (shaded). 
         \end{figuretext}
     \end{center}
\end{figure}

\begin{lemma}\label{hatw}
Let $\hat{w} = \hat{v} - I$. The following estimates hold uniformly for $t > 2$ and $\zeta \in \mathcal{I}$:
\begin{subequations}\label{hatwestimate}
\begin{align}\label{hatwestimate1}
& \|\hat{w}\|_{(L^1\cap L^2 \cap L^\infty)(\mathcal{X} \setminus \overline{D_\epsilon(\lambda_0)})} \leq C e^{-ct}, \qquad 
	\\\label{hatwestimate2}
& \|\hat{w}\|_{(L^1\cap L^2 \cap L^\infty)(\Gamma^{(3)} \setminus (\mathcal{X} \cup \partial D_\epsilon(\lambda_0))} \leq C t^{-3/2},
	\\\label{hatwestimate3}
& \|\hat{w}\|_{(L^1\cap L^2 \cap L^\infty)(\partial D_\epsilon(\lambda_0))} \leq C q |\lambda_0|^{-1/2} t^{-1/2},
	\\\label{hatwestimate4}
& \|\hat{w}\|_{L^1(\mathcal{X}^\epsilon)} \leq Ct^{-1}\ln t,
	\\\label{hatwestimate5}
& \|\hat{w}\|_{L^2(\mathcal{X}^\epsilon)} \leq Ct^{-3/4}\ln t,
	\\\label{hatwestimate6}
& \|\hat{w}\|_{L^\infty(\mathcal{X}^\epsilon)} \leq  C t^{-1/2} \ln t.
\end{align}
\end{subequations}
\end{lemma}
\begin{proof}
For $\lambda \in \mathcal{X}_2 \setminus \overline{D_\epsilon(\lambda_0)}$, only the $(12)$ entry of $\hat{w}$ is nonzero and, using (\ref{rjaestimates}),
\begin{align*}
|(\hat{w}(x,t,\lambda))_{12}| & = |-\delta(\zeta, \lambda)^2 r_{2,a}(x,t,\lambda) e^{-t\Phi(\zeta, \lambda)}|
\leq C |r_{2,a}| e^{-t|\re \Phi|}
	\\
& \leq C e^{-\frac{3t}{4}|\re \Phi|}, \qquad \lambda \in \mathcal{X}_2 \setminus \overline{D_\epsilon(\lambda_0)}.
\end{align*}
But $\re \Phi(\zeta, \lambda) = -8(\im \lambda)(\re \lambda - \lambda_0)$; hence $-|\re \Phi(\zeta, \lambda)| \leq -4\epsilon^2$ for all $\lambda \in \mathcal{X} \setminus \overline{D_\epsilon(\lambda_0)}$ and $\zeta \in \mathcal{I}$. We find
$$|(\hat{w}(x,t,\lambda))_{12}| \leq C e^{- 3\epsilon^2 t}, \qquad \lambda \in \mathcal{X}_2 \setminus \overline{D_\epsilon(\lambda_0)}, \ \zeta \in \mathcal{I}, \ t \geq 1.$$
This gives the estimate (\ref{hatwestimate1}) on $\mathcal{X}_2 \setminus \overline{D_\epsilon(\lambda_0)}$; the norms on $\mathcal{X}_j \setminus \overline{D_\epsilon(\lambda_0)}$, $j = 1,3,4$, are estimated in a similar way. This proves (\ref{hatwestimate1}).

The jump matrix $v^{(3)}$ on $\Gamma^{(3)} \setminus (\mathcal{X} \cup \partial D_\epsilon(\lambda_0))$ involves the small remainders $h_r$, $r_{1,r}$, $r_{2,r}$, so the estimate (\ref{hatwestimate2}) holds as a consequence of Lemma \ref{decompositionlemma} and Lemma \ref{decompositionlemma2}.

The estimate (\ref{hatwestimate3}) follows from (\ref{mmodmuestimate2}).
For $\lambda \in \mathcal{X} \cap D_\epsilon(\lambda_0)$, we have
$$\hat{w} = m_-^{\lambda_0} (v^{(3)} - v^{\lambda_0}) (m_+^{\lambda_0})^{-1},$$
so equations (\ref{mlambda0bound}) and (\ref{v5vmuestimate}) yield (\ref{hatwestimate4})-(\ref{hatwestimate6}).
\end{proof}

The estimates in Lemma \ref{hatw} show that
\begin{align}\label{hatwLinfty}
\begin{cases}
\|\hat{w}\|_{(L^1 \cap L^2)(\hat{\Gamma})} \leq C t^{-1/2},
	\\
\|\hat{w}\|_{L^\infty(\hat{\Gamma})} \leq  C t^{-1/2} \ln t,
\end{cases}	 \qquad t > 2, \ \zeta \in \mathcal{I}.
\end{align}
If $f \in L^2(\hat{\Gamma})$, then the Cauchy transform $\hat{\mathcal{C}}f$ is defined by
\begin{align}\label{Cauchytransform}
(\hat{\mathcal{C}}f)(\lambda) = \frac{1}{2\pi i} \int_{\hat{\Gamma}} \frac{f(s)}{s - \lambda} ds, \qquad \lambda \in \C \setminus \hat{\Gamma}.
\end{align}
We let $\hat{\mathcal{C}}_+ f$ and $\hat{\mathcal{C}}_-f$ denote the nontangential boundary values of $\hat{\mathcal{C}}f$ from the left and right sides of $\hat{\Gamma}$. Then $\hat{\mathcal{C}}_\pm \in \mathcal{B}(L^2(\hat{\Gamma}))$ and $\hat{\mathcal{C}}_+ - \hat{\mathcal{C}}_- = I$, where
$\mathcal{B}(L^2(\hat{\Gamma}))$ denotes the Banach space of bounded linear maps $L^2(\hat{\Gamma}) \to L^2(\hat{\Gamma})$.
By (\ref{hatwLinfty}), we have
\begin{align}\label{Chatwnorm}
\|\hat{\mathcal{C}}_{\hat{w}}\|_{\mathcal{B}(L^2(\hat{\Gamma}))} \leq C \|\hat{w}\|_{L^\infty(\hat{\Gamma})}  
\leq C t^{-1/2} \ln t, \qquad t > 2, \ \zeta \in \mathcal{I},
\end{align}
where the operator $\hat{\mathcal{C}}_{\hat{w}}: L^2(\hat{\Gamma}) + L^\infty(\hat{\Gamma}) \to L^2(\hat{\Gamma})$ is defined by
$\hat{\mathcal{C}}_{\hat{w}}f = \hat{\mathcal{C}}_-(f \hat{w})$.

In particular, there exists a $T > 2$ such that $I - \hat{\mathcal{C}}_{\hat{w}(\zeta, t, \cdot)} \in \mathcal{B}(L^2(\hat{\Gamma}))$ is invertible for all $t > T$. 
We define $\hat{\mu}(x, t, \lambda) \in I + L^2(\hat{\Gamma})$ for $t > T$ by
\begin{align}\label{hatmudef}
\hat{\mu} = I + (I - \hat{\mathcal{C}}_{\hat{w}})^{-1}\hat{\mathcal{C}}_{\hat{w}}I.
\end{align}
Standard estimates using the Neumann series show that
\begin{align*}
\|\hat{\mu} - I\|_{L^2(\hat{\Gamma})} \leq  
\frac{C\|\hat{w}\|_{L^2(\hat{\Gamma})}}{1 - \|\hat{\mathcal{C}}_{\hat{w}}\|_{\mathcal{B}(L^2(\hat{\Gamma}))}}.
\end{align*}
Thus, by (\ref{hatwLinfty}) and (\ref{Chatwnorm}),
\begin{align}\label{muhatestimate}
\|\hat{\mu}(x,t,\cdot) - I\|_{L^2(\hat{\Gamma})} \leq Ct^{-1/2}, \qquad t > T, \ \zeta \in \mathcal{I}.
\end{align}
It follows that there exists a unique solution $\hat{m} \in I + \dot{E}^2(\hat{\C} \setminus \hat{\Gamma})$ of the RH problem (\ref{RHmhat}) for all $t > T$. This solution is given by
\begin{align}\label{hatmrepresentation}
\hat{m}(x, t, \lambda) = I + \hat{\mathcal{C}}(\hat{\mu}\hat{w}) = I + \frac{1}{2\pi i}\int_{\hat{\Gamma}} \hat{\mu}(x, t, s) \hat{w}(x, t, s) \frac{ds}{s - \lambda}.
\end{align}

\subsection{Asymptotics of $\hat{m}$}
The following nontangential limit exists as $\lambda \to \infty$:
\begin{align}\label{limlhatm}
\ntlim_{\lambda\to \infty} \lambda(\hat{m}(x,t,\lambda) - I) 
= - \frac{1}{2\pi i}\int_{\hat{\Gamma}} \hat{\mu}(x,t,\lambda) \hat{w}(x,t,\lambda) d\lambda.
\end{align}
By (\ref{hatwestimate2}) and (\ref{muhatestimate}),
\begin{align*}
\int_{\Gamma'} \hat{\mu}(x,t,\lambda) \hat{w}(x,t,\lambda) d\lambda
& = \int_{\Gamma'} \hat{w}(x,t,\lambda) d\lambda
+
\int_{\Gamma'} (\hat{\mu}(x,t,\lambda) - I) \hat{w}(x,t,\lambda) d\lambda
	\\
& = O(\|\hat{w}\|_{L^1(\Gamma')})
+ O(\|\hat{\mu} - I\|_{L^2(\Gamma')}\|\hat{w}\|_{L^2(\Gamma')})
	\\
& = O(t^{-3/2}), \qquad t \to \infty,
\end{align*}
where $\Gamma' := \Gamma^{(3)} \setminus (\mathcal{X}^\epsilon \cup \partial D_\epsilon(\lambda_0))$.
Hence the contribution to the integral in (\ref{limlhatm}) from $\Gamma'$ is $O(t^{-3/2})$.
By (\ref{mmodmuestimate1}), (\ref{hatwestimate3}), and (\ref{muhatestimate}), the contribution from $\partial D_\epsilon(\lambda_0)$ to the right-hand side of (\ref{limlhatm}) is
\begin{align*}
&- \frac{1}{2\pi i} \int_{\partial D_\epsilon(\lambda_0)} \hat{w}(x,t,\lambda) d\lambda	
 - \frac{1}{2\pi i} \int_{\partial D_\epsilon(\lambda_0)} (\hat{\mu}(x,t,\lambda) - I) \hat{w}(x,t,\lambda) d\lambda
	\\
= & - \frac{1}{2\pi i}\int_{\partial D_\epsilon(\lambda_0)}((m^{\lambda_0})^{-1} - I) d\lambda
+ O(\|\hat{\mu} - I\|_{L^2(\partial D_\epsilon(\lambda_0))}\|\hat{w}\|_{L^2(\partial D_\epsilon(\lambda_0))})
	\\
= & \; \frac{Y(\zeta, t) m_1^X(\zeta) Y(\zeta, t)^{-1}}{\sqrt{8t}} 
+ O\bigg(\frac{q}{|\lambda_0|^{1/2} t}\bigg), \qquad t \to \infty.
\end{align*}
Finally, by (\ref{hatwestimate}) and (\ref{muhatestimate}), the contribution from $\mathcal{X}^\epsilon$ to the right-hand side of (\ref{limlhatm}) is
$$O\big(\|\hat{w}\|_{L^1(\mathcal{X}^\epsilon)}
+ \|\hat{\mu} - I\|_{L^2(\mathcal{X}^\epsilon)}\|\hat{w}\|_{L^2(\mathcal{X}^\epsilon)}\big)
= O(t^{-1} \ln t), \qquad t\to \infty.$$
Collecting the above contributions, we find from (\ref{limlhatm}) that
\begin{align}\label{limlhatm2}
\ntlim_{\lambda\to \infty} \lambda(\hat{m}(x,t,\lambda) - I) 
= \frac{Y(\zeta, t) m_1^X(\zeta) Y(\zeta, t)^{-1}}{\sqrt{8t}} 
+ O\bigg(\frac{\ln t}{t}\bigg), \qquad t\to \infty,
\end{align}
uniformly with respect to $\zeta \in \mathcal{I}$.

\subsection{Asymptotics of $u$}
Taking the transformations of Section \ref{transsec} into account, we find 
$$
m(x,t,\lambda) = \hat{m}(x,t,\lambda)H(x,t,\lambda)^{-1}\delta(\zeta, t)^{\sigma_3}$$
for all large $\lambda \in \C \setminus \hat{\Gamma}$.
It follows from (\ref{tildeurecover}), (\ref{deltaasymptotics}), (\ref{Hdef}), and (\ref{limlhatm2}) that
\begin{align}\nonumber
\tilde{u}(x,t) =& \ntlim_{\lambda\to \infty} \lambda(m(x,t,\lambda) - I)_{12} 
	\\\nonumber
= & \ntlim_{\lambda\to \infty} \lambda(\hat{m}(x,t,\lambda) - I)_{12}
	\\\nonumber
=& \; \frac{(Y(\zeta, t) m_1^X(\zeta) Y(\zeta, t)^{-1})_{12}}{\sqrt{8t}} + O(t^{-1}\ln t)
	\\\nonumber
=& - \frac{i\beta^X(q)\delta_0^2 e^{-t\Phi(\zeta, \lambda_0)}}{\sqrt{8t|\lambda_0|}}+ O(t^{-1}\ln t)
	\\ \label{utildeasymptotics}
=& - \frac{i\beta^X(q) (8t)^{i\nu} e^{2\chi(\zeta, \lambda_0)}e^{4it\lambda_0^2}}{\sqrt{8t|\lambda_0|}}+ O(t^{-1}\ln t),  \qquad t\to \infty,
\end{align} 
uniformly with respect to $\zeta \in \mathcal{I}$. 

By (\ref{urecover}), we have
\begin{align}\label{uutilde}
u(x,t) = 2i\tilde{u}(x,t)e^{2i\int^{(x,t)}_{(0, 0)} \Delta}.
\end{align}
If we choose an integration contour consisting of the vertical segment from $(0,0)$ to $(0,t)$ followed by the horizontal segment from $(0,t)$ to $(x,t)$, it follows from the definition (\ref{Deltadef}) of $\Delta$ and the equality $|u| = 2|\tilde{u}|$ that
\begin{align}\label{intDelta}
\int^{(x,t)}_{(0, 0)} \Delta
= \int_0^t \bigg(\frac{3}{4}|g_0|^4 - \frac{i}{2}(\bar{g}_1g_0 - \bar{g}_0 g_1)\bigg)dt'
+ \int_0^x  2|\tilde{u}(x',t)|^2 dx'.
\end{align}
The asymptotic formula (\ref{utildeasymptotics}) for $\tilde{u}$ gives 
\begin{align}\nonumber
\int_0^x  2|\tilde{u}(x',t)|^2 dx' 
= &\; 2\int_0^x \bigg|\sqrt{\frac{\nu(\zeta')}{8t|\lambda_0(\zeta')|}\bigg|_{\zeta' = \frac{x'}{t}}} + O(t^{-1}\ln t)\bigg|^2 dx' 
	\\\nonumber
= &\; \frac{1}{8\pi t}\int_0^x \frac{\ln(1+\frac{x'}{4t}|r(-\frac{x'}{4t})|^2 )}{x'/(4t)} dx'  + O(xt^{-3/2} \ln t )+O(xt^{-2} (\ln t)^2 )
	\\\label{inttildeu}
=&\; \frac{1}{2\pi}\int_0^{|\lambda_0|} \frac{\ln(1+s|r(-s)|^2)}{s}ds  + O(t^{-1/2} \ln t ), \qquad t \to \infty,
\end{align}
where the error term is uniform with respect to $\zeta \in \mathcal{I}$ and we have used that $x \leq Ct$ and $\nu \leq C |\lambda_0|$ for $\zeta \in \mathcal{I}$.
Substituting (\ref{utildeasymptotics}), (\ref{intDelta}), and (\ref{inttildeu}) into (\ref{uutilde}), the asymptotic formula (\ref{uasymptotics}) for $u(x,t)$ follows. This completes the proof of Theorem \ref{mainth}.

\appendix

\section{Proof of Lemma \ref{decompositionlemma}} \label{appA}
For the proof of Lemma \ref{decompositionlemma}, we will need the following properties of $h$:
\begin{itemize}
\item $h \in C^{6}(i\R_+)$, where $\R_+ = [0, \infty)$.
\item There exist complex constants $\{p_j\}_0^{5}$ such that
\begin{align*}
& h^{(n)}(\lambda) = \frac{d^n}{d{\lambda}^n}\bigg(\sum_{j=0}^{5} p_j {\lambda}^j\bigg)+O({\lambda}^{6-n}), \qquad \lambda \to 0, \ \lambda \in i\R_{+}, \ n = 0,1,2.
\end{align*}
\item There exists a complex constants $\{h_j\}_2^3$ such that
\begin{align*}
& h^{(n)}(\lambda) = \frac{d^n}{d{\lambda}^n}\bigg(\sum_{j=2}^3 \frac{h_j}{\lambda^j} \bigg)+O\Big(\frac{1}{\lambda^{4+n}}\Big), \qquad \lambda \to \infty, \ \lambda \in i\R_{+}, \ n = 0,1,2.
\end{align*}
\end{itemize}
The above properties are direct consequences of the fact that $h(\lambda)$ is a smooth function of $\lambda \in \bar{\mathcal{D}}_2$ which satisfies (\ref{hexpansion}) as $\lambda \to \infty$. 

Let
$$f_0(\lambda) = \sum_{j=2}^{9} \frac{a_j}{(\lambda + i)^j},$$
where $\{a_j\}_2^{9}$ are complex constants such that
\begin{align}\label{linearconditions}
f_0(\lambda) = 
\begin{cases}
\sum_{j=0}^{5} p_j {\lambda}^j + O({\lambda}^{6}), & \lambda \to 0,
	\\
 \sum_{j=2}^{3} h_j {\lambda}^{-j} + O({\lambda}^{-4}), & \lambda \to \infty.
\end{cases}
\end{align}
It is easy to verify that (\ref{linearconditions}) imposes eight linearly independent conditions on the $a_j$; hence the coefficients $a_j$ exist and are unique.
The function $f_0(\lambda)$ coincides with $h(\lambda)$ to fifth order at $0$ and to third order at $\infty$. More precisely, letting $f = h - f_0$, we have
\begin{align}\label{hcoincide}
 f^{(n)}(\lambda) =
\begin{cases}
 O(\lambda^{6 - n}), & \lambda \to 0, 
	\\
O({\lambda}^{-4-n}), & \lambda \to \infty, 
 \end{cases}
 \qquad  \lambda \in i\R_{+}, \ n = 0,1,2.
\end{align}

The decomposition of $h(\lambda)$ can now be derived as follows.
The map $\lambda \mapsto \phi = \phi(\lambda)$ defined by
$$\phi(\lambda) = 4 {\lambda}^2$$ 
is a bijection $[0,i\infty) \to (-\infty,0]$, so we may define a function $F:\R \to \C$ by
\begin{align}\label{Fphi}
F(\phi) = \begin{cases} \frac{(\lambda +i)^3}{\lambda} f(\lambda), & \phi \leq 0, \\
0, & \phi > 0, 
\end{cases}
\quad \zeta \in \mathcal{I},
\end{align}
The function $F(\phi)$ is $C^6$ for $\phi \neq 0$ and 
$$F^{(n)}(\phi) = \bigg(\frac{1}{8 \lambda} \frac{\partial }{\partial \lambda}\bigg)^n \bigg(\frac{(\lambda+i)^3}{\lambda}f(\lambda)\bigg), \qquad \phi < 0.$$ 
Using (\ref{hcoincide}) it follows that $F \in C^2(\R)$ and $F^{(n)}(\phi) = O(|\phi|^{-1-n})$ as $|\phi| \to \infty$ for $n = 0,1,2$.
In particular,
\begin{align*}
\bigg\|\frac{d^nF}{d \phi^n}\bigg\|_{L^2(\R)} < \infty, \qquad n = 0,1,2,
\end{align*}
that is, $F$ belongs to the Sobolev space $H^2(\R)$.
It follows that the Fourier transform $\hat{F}(s)$ defined by
$$\hat{F}(s) = \frac{1}{2\pi} \int_{\R} F(\phi) e^{-i\phi s} d\phi,$$
satisfies
\begin{align}\label{FphihatF}
F(\phi) =  \int_{\R} \hat{F}(s) e^{i\phi s} ds,
\end{align}
and, by the Plancherel theorem, $\|s^2 \hat{F}(s)\|_{L^2(\R)} < \infty$.
Equations (\ref{Fphi}) and (\ref{FphihatF}) imply
$$ \frac{\lambda}{(\lambda+i)^3}\int_{\R} \hat{F}(s) e^{4is\lambda^2} ds 
= f(\lambda), \qquad  \lambda \in i\R_{+}.$$
Writing
$$f(\lambda) = f_a(t, \lambda) + f_r(t, \lambda), \qquad t > 0, \  \lambda \in i\R_{+},$$
where the functions $f_a$ and $f_r$ are defined by
\begin{align*}
& f_a(t,\lambda) = \frac{\lambda}{(\lambda+i)^3}\int_{-\frac{t}{4}}^\infty \hat{F}(s) e^{4is{\lambda}^2} ds, \qquad t > 0, \ \lambda \in \bar{{\mathcal D}}_1,  
	\\
& f_r(t,\lambda) = \frac{\lambda}{(\lambda+i)^3}\int_{-\infty}^{-\frac{t}{4}} \hat{F}(s) e^{4is{\lambda}^2} ds,\qquad t > 0, \   \lambda \in i\R_{+},
\end{align*}
we infer that $f_a(t, \cdot)$ is continuous in $\bar{{\mathcal D}}_1$ and analytic in ${\mathcal D}_1$. 
Furthermore, since $|\re 4i{\lambda}^2| \leq |\re \Phi(\zeta, \lambda)|$ for all $\lambda \in \bar{{\mathcal D}}_1$ and $\zeta \in \mathcal{I}$, we find
\begin{align*}\nonumber
 |f_a(t, \lambda)| 
&\leq \frac{|\lambda|}{|\lambda+i|^3}\|\hat{F}\|_{L^1(\R)}  \sup_{s \geq -\frac{t}{4}} e^{s \re 4i{\lambda^2}}
\leq \frac{C |\lambda|}{|\lambda+i|^3} e^{\frac{t}{4} |\re 4i{\lambda}^2|} 
	\\ 
 &\leq \frac{C |\lambda|}{1 + |\lambda|^3} e^{\frac{t}{4} |\re \Phi(\zeta, \lambda)|}, \qquad \zeta \in \mathcal{I}, \ t > 0, \ \lambda \in \bar{{\mathcal D}}_1,
\end{align*}
and
\begin{align*}\nonumber
|f_r(t, \lambda)| & \leq \frac{|\lambda|}{|\lambda+i|^3} \int_{-\infty}^{-\frac{t}{4}} s^2 |\hat{F}(s)| s^{-2} ds
 \leq \frac{C|\lambda|}{|\lambda+i|^3}  \| s^2 \hat{F}(s)\|_{L^2(\R)} \sqrt{\int_{-\infty}^{-\frac{t}{4}} s^{-4} ds}  
 	\\ 
&  \leq \frac{C}{1 + |\lambda|^2}  t^{-3/2}, \qquad \zeta \in \mathcal{I}, \ t > 0, \ \lambda \in i\R_{+}.
\end{align*}
Hence the $L^1$, $L^2$, and $L^\infty$ norms of $f_r$ on $i\R_{+}$ are $O(t^{-3/2})$. 
Letting
\begin{align*}
& h_{a}(t, \lambda) = f_0(\lambda) + f_a(t, \lambda), \qquad t > 0, \ \lambda \in \bar{{\mathcal D}}_1,
	\\
& h_{r}(t, \lambda) = f_r(t, \lambda), \qquad t > 0, \ \lambda \in i\R_{+},
\end{align*}
we find a decomposition of $h$ with the properties listed in the statement of the lemma.

\bigskip
\noindent
{\bf Acknowledgement}  {\it The authors are grateful to the referees for helpful remarks. L. K. Arruda thanks the members of the Department of Mathematics at KTH Royal Institute of Technology for their kind hospitality. J. Lenells acknowledges support from European Research Council, Consolidator Grant No. 682537, the Swedish Research Council, Grant No. 2015-05430, and the G\"oran Gustafsson Foundation, Sweden.}

\bibliographystyle{plain}
\bibliography{is}

\end{document}